\documentclass[11pt]{article}

\usepackage{amsmath,amsfonts,amssymb,amsthm,braket}
\usepackage[margin=1in]{geometry}
\usepackage[colorlinks,
            linkcolor = blue,
            urlcolor  = blue,
            citecolor = blue,
            anchorcolor = blue]{hyperref}
\usepackage{mathrsfs}
\usepackage{mathtools}
\usepackage{thmtools}
\usepackage{thm-restate}
\mathtoolsset{centercolon}
\usepackage[numbers,comma,sort&compress]{natbib}
\usepackage{bm}
\usepackage{bookmark}
\usepackage{svg}
\usepackage[font=small]{caption}
\usepackage{titlesec}
\usepackage{enumitem}

\titlespacing*{\paragraph}{0pt}{3pt plus 1pt minus 1pt}{1em}

\usepackage{enumerate}

\let\originalleft\left
\let\originalright\right
\renewcommand{\left}{\mathopen{}\mathclose\bgroup\originalleft}
\renewcommand{\right}{\aftergroup\egroup\originalright}

\newcommand{\C}{{\mathbb{C}}}

\newcommand{\R}{{\mathbb{R}}}

\renewcommand{\d}{\mathrm{d}}

\renewcommand{\Re}{\mathop{\mathrm{Re}}}
\renewcommand{\Im}{\mathop{\mathrm{Im}}}

\renewcommand{\vec}{\mathbf}

\newcommand{\norm}[1]{\|{#1}\|}

\newcommand{\range}[1]{[{#1}]}
\newcommand{\rangez}[1]{[{#1}]_0}

\DeclareMathOperator{\poly}{poly}

\DeclareMathOperator{\Rq}{R}

\newcommand{\dimn}{\Delta}

\newcommand{\rtra}{r_{\text{tra}}}
\newcommand{\Tlat}{T_{\text{lat}}}
\newcommand{\Tinf}{T_{\text{inf}}}

\newcommand{\rvac}{r_{\text{vac}}}

\newtheorem{theorem}{Theorem}
\newtheorem{lemma}{Lemma}
\newtheorem{corollary}{Corollary}
\newtheorem{problem}{Problem}

\numberwithin{equation}{section}

\newcommand{\eq}[1]{(\ref{eq:#1})}
\renewcommand{\sec}[1]{\hyperref[sec:#1]{Section~\ref*{sec:#1}}}
\newcommand{\thm}[1]{\hyperref[thm:#1]{Theorem~\ref*{thm:#1}}}
\newcommand{\lem}[1]{\hyperref[lem:#1]{Lemma~\ref*{lem:#1}}}
\newcommand{\cor}[1]{\hyperref[cor:#1]{Corollary~\ref*{cor:#1}}}
\newcommand{\prb}[1]{\hyperref[prb:#1]{Problem~\ref*{prb:#1}}}
\newcommand{\fig}[1]{\hyperref[fig:#1]{Figure~\ref*{fig:#1}}}

\makeatletter
\let\@@magyar@captionfix\relax
\makeatother

\begin{document}

\title{
Efficient quantum algorithm for 
\\ dissipative nonlinear differential equations}

\author{Jin-Peng Liu$^{1,2,3}$ \quad Herman Øie Kolden$^{4,5}$ \quad Hari K.\ Krovi$^{6}$ \\
Nuno F.\ Loureiro$^{5}$ \quad Konstantina Trivisa$^{3,7}$ \quad Andrew M.\ Childs$^{1,2,8}$ \\ [10pt]
\footnotesize $^{1}$ Joint Center for Quantum Information and Computer Science, University of Maryland, MD 20742, USA \\
\footnotesize $^{2}$ Institute for Advanced Computer Studies, University of Maryland, MD 20742, USA \\
\footnotesize $^{3}$ Department of Mathematics, University of Maryland, MD 20742, USA \\
\footnotesize $^{4}$ Department of Physics, Norwegian University of Science and Technology, NO-7491 Trondheim, Norway \\
\footnotesize $^{5}$ Plasma Science and Fusion Center, Massachusetts Institute of Technology, MA 02139, USA \\
\footnotesize $^{6}$ Raytheon BBN Technologies, MA 02138, USA \\
\footnotesize $^{7}$ Institute for Physical Science and Technology, University of Maryland, MD 20742, USA \\
\footnotesize $^{8}$ Department of Computer Science, University of Maryland, MD 20742, USA \\
}

\date{}
\maketitle

\begin{abstract}
Nonlinear differential equations model diverse phenomena but are notoriously difficult to solve.
While there has been extensive previous work on efficient quantum algorithms for linear differential equations, the linearity of quantum mechanics has limited analogous progress for the nonlinear case. Despite this obstacle, we develop a quantum algorithm for dissipative quadratic $n$-dimensional ordinary differential equations. Assuming $\Rq < 1$, where $\Rq$ is a parameter characterizing the ratio of the nonlinearity and forcing to the linear dissipation, this algorithm has complexity $T^2 q \poly(\log T, \log n, \log 1/\epsilon)/\epsilon$, where $T$ is the evolution time, $\epsilon$ is the allowed error, and $q$ measures decay of the solution. This is an exponential improvement over the best previous quantum algorithms, whose complexity is exponential in $T$. While exponential decay precludes efficiency, driven equations can avoid this issue despite the presence of dissipation. Our algorithm uses the method of Carleman linearization, for which we give a novel convergence theorem. This method maps a system of nonlinear differential equations to an infinite-dimensional system of linear differential equations, which we discretize, truncate, and solve using the forward Euler method and the quantum linear system algorithm. We also provide a lower bound on the worst-case complexity of quantum algorithms for general quadratic differential equations, showing that the problem is intractable for $\Rq \ge \sqrt{2}$. Finally, we discuss potential applications, showing that the $\Rq < 1$ condition can be satisfied in realistic epidemiological models and giving numerical evidence that the method may describe a model of fluid dynamics even for larger values of $\Rq$.
\end{abstract}
        
\section{Introduction}
\label{sec:intro}

Models governed by both ordinary differential equations  (ODEs) and partial differential equations (PDEs) arise extensively in natural and social science, medicine, and engineering. Such equations characterize physical and biological systems that exhibit a wide variety of complex phenomena, including turbulence and chaos.

We focus here on differential equations with nonlinearities that can be expressed with quadratic polynomials. Note that polynomials of degree higher than two, and even more general nonlinearities, can be reduced to the quadratic case by introducing additional variables \cite{Ker81,FP17}. The quadratic case also directly includes many archetypal models, such as the logistic equation in biology, the Lorenz system in atmospheric dynamics, and the Navier--Stokes equations in fluid dynamics.

Quantum algorithms offer the prospect of rapidly characterizing the solutions of high-dimensional systems of linear ODEs \cite{Ber14,BCOW17,CL19} and PDEs \cite{CJS13,CPPTK13,MP16,CJO19,CLO20,ESP19,LMS20}. Such algorithms can produce a quantum state proportional to the solution of a sparse (or block-encoded) $n$-dimensional system of linear differential equations in time $\poly(\log n)$ using the quantum linear system algorithm \cite{HHL09}.

Early work on quantum algorithms for differential equations already considered the nonlinear case \cite{LO08}. It gave a quantum algorithm for ODEs that simulates polynomial nonlinearities by storing multiple copies of the solution.  The complexity of this approach is polynomial in the logarithm of the dimension but exponential in the evolution time, scaling as $O(1/\epsilon^{T})$ due to exponentially increasing resources used to maintain sufficiently many copies of the solution to represent the nonlinearity throughout the evolution.

Recently, heuristic quantum algorithms for nonlinear ODEs have been studied. Reference \cite{Jos20} explores a linearization technique known as the Koopman--von Neumann method that might be amenable to the quantum linear system algorithm.  In \cite{DS20}, the authors provide a high-level description of how linearization can help solve nonlinear equations on a quantum computer.  However, neither paper makes precise statements about concrete implementations or running times of quantum algorithms. The recent preprint \cite{LPG20} also describes a quantum algorithm to solve a nonlinear ODE by linearizing it using a different approach from the one taken here. However, a proof of correctness of their algorithm involving a bound on the condition number and probability of success is not given. The authors also do not describe how barriers such as those of \cite{CY16} could be avoided in their approach.

While quantum mechanics is described by linear dynamics, possible nonlinear modifications of the theory have been widely studied. Generically, such modifications enable quickly solving hard computational problems (e.g., solving unstructured search among $n$ items in time $\poly(\log n)$), making nonlinear dynamics exponentially difficult to simulate in general \cite{AL98,Aar05,CY16}. Therefore, constructing efficient quantum algorithms for general classes of nonlinear dynamics has been considered largely out of reach.

In this article, we design and analyze a quantum algorithm that overcomes this limitation using Carleman linearization \cite{Car32,KS91,FP17}. This approach embeds polynomial nonlinearities into an infinite-dimensional system of linear ODEs, and then truncates it to obtain a finite-dimensional linear approximation. The Carleman method has previously been used in the analysis of dynamical systems \cite{SW80,And82, LT87} and the design of control systems \cite{Bro14,RMA09,GMP05}, but to the best of our knowledge it has not been employed in the context of quantum algorithms. We discretize the finite ODE system in time using the forward Euler method and solve the resulting linear equations with the quantum linear system algorithm \cite{HHL09,CKS15}.
We control the approximation error of this approach by combining a novel convergence theorem with a bound for the global error of the Euler method. Furthermore, we provide an upper bound for the condition number of the linear system and lower bound the success probability of the final measurement. Subject to the condition $\Rq<1$, where the quantity $\Rq$ (defined in \prb{node} below) characterizes the relative strength of the nonlinear and dissipative linear terms, we show that the total complexity of this \emph{quantum Carleman linearization algorithm} is $sT^2q\poly(\log T, \log n, \log 1/\epsilon)/\epsilon$, where $s$ is the sparsity, $T$ is the evolution time, $q$ quantifies the decay of the final solution relative to the initial condition, $n$ is the dimension, and $\epsilon$ is the allowed error (see \thm{main}). In the regime $\Rq<1$, this is an exponential improvement over \cite{LO08}, which has complexity exponential in $T$.

Note that the solution cannot decay exponentially in $T$ for the algorithm to be efficient, as captured by the dependence of the complexity on $q$---a known limitation of quantum ODE algorithms \cite{BCOW17}. For homogeneous ODEs with $R<1$, the solution necessarily decays exponentially in time (see equation \eq{homo_example}), so the algorithm is not asymptotically efficient. However, even for solutions with exponential decay, we still find an improvement over the best previous result $O(1/\epsilon^{T})$ \cite{LO08} for sufficiently small $\epsilon$. Thus our algorithm might provide an advantage over classical computation for studying evolution for short times. More significantly, our algorithm can handle inhomogeneous quadratic ODEs, for which it can remain efficient in the long-time limit since the solution can remain asymptotically nonzero (for an explicit example, see the discussion just before the proof of \lem{Carleman}), or can decay slowly (i.e., $q$ can be $\poly(T)$).
Inhomogeneous equations arise in many applications, including for example the discretization of PDEs with nontrivial boundary conditions. 

We also provide a quantum lower bound for the worst-case complexity of simulating strongly nonlinear dynamics, showing that the algorithm's condition $\Rq<1$ cannot be significantly improved in general (\thm{lower}). Following the approach of \cite{AL98,CY16}, we construct a protocol for distinguishing two states of a qubit driven by a certain quadratic ODE. Provided $\Rq\ge\sqrt{2}$, this procedure distinguishes states with overlap $1-\epsilon$ in time $\poly(\log(1/\epsilon))$. Since nonorthogonal quantum states are hard to distinguish, this implies a lower bound on the complexity of the quantum ODE problem.

Our quantum algorithm could potentially be applied to study models governed by quadratic ODEs arising in biology and epidemiology as well as in fluid and plasma dynamics. In particular, the celebrated Navier--Stokes equation with linear damping, which describes many physical phenomena, can be treated by our approach provided the Reynolds number is sufficiently small.
We also note that while the formal validity of our arguments assumes $R<1$, we find in one numerical experiment that our proposed approach remains valid for larger $R$ (see \sec{application}).

We emphasize that, as in related quantum algorithms for linear algebra and differential equations, instantiating our approach requires an implicit description of the problem that allows for efficient preparation of the initial state and implementation of the dynamics. Furthermore, since the output is encoded in a quantum state, readout is restricted to features that can be revealed by efficient quantum measurements. More work remains to understand how these methods might be applied, as we discuss further in \sec{discussion}.

The remainder of this paper is structured as follows. \sec{qode} introduces the quantum quadratic ODE problem. \sec{linearization} presents the Carleman linearization procedure and describes its performance. \sec{proof} gives a detailed analysis of the quantum Carleman linearization algorithm. \sec{lower} establishes a quantum lower bound for simulating quadratic ODEs. \sec{application} describes how our approach could be applied to several well-known ODEs and PDEs and presents numerical results for the case of the viscous Burgers equation. Finally, we conclude with a discussion of the results and some possible future directions in \sec{discussion}.

\section{Quadratic ODEs}
\label{sec:qode}
We focus on an initial value problem described by the $n$-dimensional quadratic ODE
\begin{equation}
\frac{\d{u}}{\d{t}} = F_2u^{\otimes 2}+F_1u+F_0(t), \qquad
u(0) = u_{\mathrm{in}}.
\label{eq:NODE}
\end{equation}
Here $u=[u_1, \ldots, u_n]^{T}\in\R^n$, $u^{\otimes 2}=[u_1^2, u_1u_2, \ldots, u_1u_n, u_2u_1, \ldots, u_nu_{n-1}, u_n^2]^{T}\in\R^{n^2}$, each $u_j = u_j(t)$ is a function of $t$ on the interval $[0,T]$ for $j\in\range{n}\coloneqq\{1,\ldots,n\}$, $F_2\in\R^{n\times n^2}$, $F_1\in\R^{n\times n}$ are time-independent matrices, and the inhomogeneity $F_0(t)\in\R^n$ is a $C^1$ continuous function of $t$. We let $\|\cdot\|$ denote the spectral norm.

The main computational problem we consider is as follows.
\begin{problem}\label{prb:node}
In the \emph{quantum quadratic ODE problem}, we consider an $n$-dimensional quadratic ODE as in \eq{NODE}. We assume $F_2$, $F_1$, and $F_0$ are $s$-sparse (i.e., have at most $s$ nonzero entries in each row and column), $F_1$ is diagonalizable, and that the eigenvalues $\lambda_j$ of $F_1$ satisfy $\Re{(\lambda_n)} \le \cdots \le \Re{(\lambda_1)} < 0$. We parametrize the problem in terms of the quantity
\begin{equation}
\Rq \coloneqq \frac{1}{|\Re{(\lambda_1)}|}\biggl(\|u_{\mathrm{in}}\|\|F_2\|+\frac{\|F_0\|}{\|u_{\mathrm{in}}\|}\biggr).
\label{eq:A1}
\end{equation}
For some given $T>0$, we assume the values $\Re{(\lambda_1)}$, $\|F_2\|$, $\|F_1\|$, $\|F_0(t)\|$ for each $t\in[0,T]$, and $\|F_0\| \coloneqq \max_{t\in[0,T]} \|F_0(t)\|$, $\|F_0'\| \coloneqq \max_{t\in[0,T]} \|F'_0(t)\|$ are known, and that we are given oracles $O_{F_2}$, $O_{F_1}$, and $O_{F_0}$ that provide the locations and values of the nonzero entries of $F_2$, $F_1$, and $F_0(t)$ for any specified $t$, respectively, for any desired row or column. We are also given the value $\|u_{\mathrm{in}}\|$ and an oracle $O_x$ that maps $|00\ldots0\rangle \in\C^n$ to a quantum state proportional to $u_{\mathrm{in}}$. 
Our goal is to produce a quantum state proportional to $u(T)$ for some given $T>0$ within some prescribed error tolerance $\epsilon>0$.
\end{problem}

When $F_0(t)=0$ (i.e., the ODE is homogeneous), the quantity $\Rq = \frac{\|u_{\mathrm{in}}\|\|F_2\|}{|\Re{(\lambda_1)}|}$ is qualitatively similar to the Reynolds number, which characterizes the ratio of the (nonlinear) convective forces to the (linear) viscous forces within a fluid \cite{Bur48,Lem18}.
More generally, $\Rq$ quantifies the combined strength of the nonlinearity and the inhomogeneity relative to dissipation.

Note that without loss of generality, given a quadratic ODE satisfying \eq{NODE} with $\Rq<1$, we can modify it by rescaling $u \to \gamma u$ with a suitable constant $\gamma$ to satisfy
\begin{equation}
\|F_2\| + \|F_0\| < |{\Re{(\lambda_1)}}|
\label{eq:A2}
\end{equation}
and
\begin{equation}
\|u_{\mathrm{in}}\| < 1,
\label{eq:A3}
\end{equation}
with $\Rq$ left unchanged by the rescaling. We use this rescaling in our algorithm and its analysis. With this rescaling, a small $\Rq$ implies both small $\|u_{\mathrm{in}}\|\|F_2\|$ and small $|F_0\|/\|u_{\mathrm{in}}\|$ relative to $|{\Re{(\lambda_1)}}|$. Our analysis also relies on the assumption that $\|F_0\| \le \|F_2\|$.

\section{Quantum Carleman linearization}
\label{sec:linearization}

As discussed in \sec{intro}, it is challenging to directly simulate quadratic ODEs using quantum computers, and indeed the complexity of the best known quantum algorithm is exponential in the evolution time $T$ \cite{LO08}.
However, for a dissipative nonlinear ODE without a source, any quadratic nonlinear effect will only be significant for a finite time because of the dissipation. To exploit this, we can create a linear system that approximates the initial nonlinear evolution within some controllable error. After the nonlinear effects are no longer important, the linear system properly captures the almost-linear evolution from then on.

We develop such a quantum algorithm using the concept of Carleman linearization \cite{Car32,KS91,FP17}. Carleman linearization is a method for converting a finite-dimensional system of nonlinear differential equations into an infinite-dimensional linear one. This is achieved by introducing powers of the variables into the system, allowing it to be written as an infinite sequence of coupled linear differential equations. We then truncate the system to $N$ equations, where the truncation level $N$ depends on the allowed approximation error, giving a finite linear ODE system. 

Let us describe the Carleman linearization procedure in more detail. 
Given a system of quadratic ODEs \eq{NODE}, we apply the Carleman procedure to obtain the system of linear ODEs
\begin{equation}
  \frac{\d{\hat y}}{\d{t}} = A(t) \hat y + b(t), \qquad
  \hat y(0) = \hat y_{\mathrm{in}}
\label{eq:LODE}
\end{equation}
with the tri-diagonal block structure
\begin{equation}
\frac{\d{}}{\d{t}}
  \begin{pmatrix}
    \hat y_1 \\
    \hat y_2 \\
    \hat y_3 \\
    \vdots \\
    \hat y_{N-1} \\
    \hat y_N \\
  \end{pmatrix}
=
  \begin{pmatrix}
    A_1^1 & A_2^1 &  &  &  &  \\
    A_1^2 & A_2^2 & A_3^2 & &  &  \\
     & A_2^3 & A_3^3 & A_4^3 &  &  \\
     &  & \ddots & \ddots & \ddots &  \\
     &  &  & A_{N-2}^{N-1} & A_{N-1}^{N-1} & A_N^{N-1} \\
     &  &  &  & A_{N-1}^N & A_N^N \\
  \end{pmatrix}
  \begin{pmatrix}
    \hat y_1 \\
    \hat y_2 \\
    \hat y_3 \\
    \vdots \\
    \hat y_{N-1} \\
    \hat y_N \\
  \end{pmatrix}+
  \begin{pmatrix}
    F_0(t) \\
    0 \\
    0 \\
    \vdots \\
    0 \\
    0 \\
  \end{pmatrix},
\label{eq:UODE}
\end{equation}
where $\hat y_j=u^{\otimes j}\in\R^{n^j}$, $\hat y_{\mathrm{in}}=[u_{\mathrm{in}}; u_{\mathrm{in}}^{\otimes 2}; \ldots; u_{\mathrm{in}}^{\otimes N}]$, and $A_{j+1}^j \in \R^{n^j\times n^{j+1}}$, $A_j^j \in \R^{n^j\times n^j}$, $A_{j-1}^j \in \R^{n^j\times n^{j-1}}$ for $j\in\range{N}$ satisfying
\begin{align}
A_{j+1}^j &= F_2\otimes I^{\otimes j-1}+I\otimes F_2\otimes I^{\otimes j-2}+\cdots+I^{\otimes j-1}\otimes F_2, \label{eq:tensor2} \\
A_j^j &= F_1\otimes I^{\otimes j-1}+I\otimes F_1\otimes I^{\otimes j-2}+\cdots+I^{\otimes j-1}\otimes F_1, \label{eq:tensor1} \\
A_{j-1}^j &= F_0(t)\otimes I^{\otimes j-1}+I\otimes F_0(t)\otimes I^{\otimes j-2}+\cdots+I^{\otimes j-1}\otimes F_0(t). \label{eq:tensor0}
\end{align}
Note that $A$ is a $(3Ns)$-sparse matrix. The dimension of \eq{LODE} is
\begin{equation}
  \dimn \coloneqq n+n^2+\cdots+n^N=\frac{n^{N+1}-n}{n-1}=O(n^N).
\end{equation}

To construct a system of linear equations, we next divide $[0,T]$ into $m = T/h$ time steps and apply the forward Euler method on \eq{LODE}, letting 
\begin{equation}
y^{k+1} = [I+A(kh)h] y^k + b(kh)
\label{eq:forward}
\end{equation}
where $y^k \in \R^{\dimn}$ approximates $\hat y(kh)$ for each $k \in \rangez{m+1} \coloneqq \{0,1,\ldots,m\}$, with $y^0 = y_{\mathrm{in}} \coloneqq \hat y(0) = \hat y_{\mathrm{in}}$, and letting all $y^k$ be equal for $k\in\rangez{m+p+1}\setminus\rangez{m+1}$, for some sufficiently large integer $p$.
(It is unclear whether another discretization could improve performance, as discussed further in \sec{discussion}.)
This gives an $(m+p+1)\dimn\times(m+p+1)\dimn$ linear system
\begin{equation}
L|Y\rangle=|B\rangle
\label{eq:linear_system}
\end{equation}
that encodes \eq{forward} and uses it to produce a numerical solution at time $T$, where
\begin{equation}
L = \sum_{k=0}^{m+p}|k\rangle\langle k|\otimes I-\sum_{k=1}^{m}|k\rangle\langle k-1|\otimes [I+A((k-1)h)h]-\sum_{h=m+1}^{m+p}|k\rangle\langle k-1|\otimes I
\label{eq:matrixL}
\end{equation}
and
\begin{equation}
|B\rangle = \frac{1}{\sqrt{B_m}}\Bigl(\|y_{\mathrm{in}}\||0\rangle \otimes |y_{\mathrm{in}}\rangle+\sum_{k=1}^{m}\|b((k-1)h)\||k\rangle \otimes |b((k-1)h)\rangle\Bigr)
\label{eq:vectorB}
\end{equation}
with a normalizing factor $B_m$. Observe that the system \eq{linear_system} has the lower triangular structure
\renewcommand{\arraystretch}{1.4}
\begin{equation}
  \begin{pmatrix}
    I &  &  &  &  &  &  \\
    -[I\!+\!A(0)h] & I &  &  &  &  &  \\
     & \ddots & \ddots &  &  &  & \\
     &  & -[I\!+\!A((m-1)h)h] & I &  &  &  \\
     &  &  & -I & I &  &  \\
     &  &  &  & \ddots & \ddots & \\
     &  &  &  &  & -I & I \\
  \end{pmatrix}
  \begin{pmatrix}
    y^0 \\
    y^1 \\
    \vdots \\
    y^m \\
    y^{m+1} \\
    \vdots \\
    y^{m+p} \\
  \end{pmatrix}
=
  \begin{pmatrix}
    y_{\mathrm{in}} \\
    b(0) \\
    \vdots \\
    b((m-1)h) \\
    0 \\
    \vdots \\
    0 \\
  \end{pmatrix}.
\end{equation}
\renewcommand{\arraystretch}{1}

In the above system, the first $n$ components of $y^k$ for $k\in\rangez{m+p+1}$ (i.e., $y_1^k$) approximate the exact solution $u(T)$, up to normalization. We apply the high-precision quantum linear system algorithm (QLSA) \cite{CKS15} to \eq{linear_system} and postselect on $k$ to produce $y_1^k/\|y^k_1\|$ for some $k\in\rangez{m+p+1}\setminus\rangez{m}$. The resulting error is at most
\begin{equation}
\epsilon \coloneqq \max_{k\in\rangez{m+p+1}\setminus\rangez{m}}\biggl\|\frac{u(T)}{\|u(T)\|}-\frac{y^k_1}{\|y^k_1\|}\biggr\|.
\label{eq:error}
\end{equation}
This error includes contributions from both Carleman linearization and the forward Euler method. (The QLSA also introduces error, which we bound separately. Note that we could instead apply the original QLSA \cite{HHL09} instead of its subsequent improvement \cite{CKS15}, but this would slightly complicate the error analysis and might perform worse in practice.)

Now we state our main algorithmic result.

\begin{theorem}[Quantum Carleman linearization algorithm]\label{thm:main}
Consider an instance of the quantum quadratic ODE problem as defined in \prb{node}, with its Carleman linearization as defined in \eq{LODE}. Assume $\Rq<1$ and $\|F_0\| \le \|F_2\|$. Let
\begin{align}
g &\coloneqq \|u(T)\|, &
q &\coloneqq \frac{\|u_{\mathrm{in}}\|}{\|u(T)\|}.
\label{eq:gq}
\end{align}
There exists a quantum algorithm producing a state that approximates $u(T)/\| u(T)\|$ with error at most $\epsilon\le1$, succeeding with probability $\Omega(1)$, with a flag indicating success, using at most
\begin{equation}
\begin{aligned}
\frac{sT^2q[(\|F_2\|+\|F_1\|+\|F_0\|)^2+\|F_0'\|]}{(1-\|u_{\mathrm{in}}\|)^2(\|F_2\|+\|F_0\|)g\epsilon}\cdot
\poly\biggl(\log\biggl(\frac{sT\|F_2\|\|F_1\|\|F_0\|\|F'_0\|}{(1-\|u_{\mathrm{in}}\|)g\epsilon}\biggr)/\log (1/\|u_{\mathrm{in}}\|) \biggr)
\end{aligned}
\end{equation}
queries to the oracles $O_{F_2}$, $O_{F_1}$, $O_{F_0}$ and $O_{x}$. 
The gate complexity is larger than the query complexity by a factor of $\poly\bigl(\log(nsT\|F_2\|\|F_1\|\|F_0\|\|F'_0\|/(1-\|u_{\mathrm{in}}\|) g\epsilon)/\log (1/\|u_{\mathrm{in}}\|)\bigr)$.
Furthermore, if the eigenvalues of $F_1$ are all real, the query complexity is
\begin{equation}
\begin{aligned}
\frac{sT^2q[(\|F_2\|+\|F_1\|+\|F_0\|)^2+\|F_0'\|]}{g\epsilon}\cdot
\poly\biggl(\log\biggl(\frac{sT\|F_2\|\|F_1\|\|F_0\|\|F'_0\|}{g\epsilon}\biggr)/\log (1/\|u_{\mathrm{in}}\|) \biggr)
\end{aligned}
\end{equation}
and the gate complexity is larger by a factor of $\poly\bigl(\log(nsT\|F_2\|\|F_1\|\|F_0\|\|F'_0\|/g\epsilon)/\log (1/\|u_{\mathrm{in}}\|)\bigr)$.
\end{theorem}

\section{Algorithm analysis}
\label{sec:proof}

In this section we establish several lemmas and use them to prove \thm{main}.

\subsection{Solution error}

The solution error has three contributions: the error from applying Carleman linearization to \eq{NODE}, the error in the time discretization of \eq{LODE} by the forward Euler method, and the error from the QLSA. Since the QLSA produces a solution with error at most $\epsilon$ with complexity $\poly(\log(1/\epsilon))$ \cite{CKS15}, we focus on bounding the first two contributions.

\subsubsection{Error from Carleman linearization}
\label{sec:carleman_error}

First, we provide an upper bound for the error from Carleman linearization for arbitrary evolution time $T$. To the best of our knowledge, the first and only explicit bound on the error of Carleman linearization appears in \cite{FP17}. However, they only consider homogeneous quadratic ODEs; and furthermore, to bound the error for arbitrary $T$, they assume the logarithmic norm of $F_1$ is negative (see Theorems 4.2 and 4.3 of \cite{FP17}), which is too strong for our case. Instead, we give a novel analysis under milder conditions, providing the first convergence guarantee for general inhomogeneous quadratic ODEs.

We begin with a lemma that describes the decay of the solution of \eq{NODE}.

\begin{lemma}\label{lem:decrease}
Consider an instance of the quadratic ODE \eq{NODE}, and assume $\Rq < 1$ as defined in \eq{A1}. Let
\begin{equation}
r_{\pm} \coloneqq \frac{-\Re(\lambda_1)\pm\sqrt{\Re(\lambda_1)^2-4\|F_2\|\|F_0\|}}{2\|F_2\|}.
\label{eq:roots}
\end{equation}
Then $r_\pm$ are distinct real numbers with $0\le r_-<r_+$, and the solution $u(t)$ of \eq{NODE} satisfies
$\|u(t)\|<\|u_{\mathrm{in}}\|<r_+$
for any $t>0$.
\end{lemma}

\begin{proof}
Consider the derivative of $\|u(t)\|$.
We have
\begin{align}
    \frac{\d\|u\|^2}{\d{t}}&= u^\dag F_2(u\otimes u) + (u^\dag\otimes u^\dag)F_2^\dag u + u^\dag(F_1 + F_1^\dag)u + u^\dag F_0(t) + F_0(t)^\dag u,\nonumber\\
    &\leq  2\|F_2\|\|u\|^3 + 2\Re(\lambda_1)\|u\|^2 + 2\|F_0\|\|u\|.
\end{align}
If $\|u\|\neq 0$, then
\begin{equation}
    \frac{\d\|u\|}{\d{t}}\leq \|F_2\|\|u\|^2 + \Re(\lambda_1)\|u\| + \|F_0\|.
\end{equation}
Letting $a=\|F_2\|>0$, $b=\Re(\lambda_1)<0$, and $c=\|F_0\|>0$ with $a,b,c>0$, we consider a $1$-dimensional quadratic ODE
\begin{equation}
\frac{\d{x}}{\d{t}} = ax^2+bx+c, \qquad
x(0) = \|u_{\mathrm{in}}\|.
\end{equation}
Since $\Rq<1 \Leftrightarrow b>a\|u_{\mathrm{in}}\|+\frac{c}{\|u_{\mathrm{in}}\|}$, the discriminant satisfies
\begin{equation}
b^2-4ac > \biggl(a\|u_{\mathrm{in}}\|+\frac{c}{\|u_{\mathrm{in}}\|}\biggr)^2-4a\|u_{\mathrm{in}}\|\cdot\frac{c}{\|u_{\mathrm{in}}\|} \ge \biggl(a\|u_{\mathrm{in}}\|-\frac{c}{\|u_{\mathrm{in}}\|}\biggr)^2 \ge 0.
\end{equation}
Thus, $r_\pm$ defined in \eq{roots} are distinct real roots of $ax^2+bx+c$. Since $r_-+r_+ = -\frac{b}{a} > 0$ and $r_-r_+ = \frac{c}{a} \ge 0$, we have $0 \le r_- < r_+$. We can rewrite the ODE as
\begin{equation}
\frac{\d{x}}{\d{t}} = ax^2+bx+c = a(x-r_-)(x-r_+), \qquad
x(0) = \|u_{\mathrm{in}}\|.
\end{equation}
Letting $y=x-r_-$, we obtain an associated homogeneous quadratic ODE
\begin{equation}
\frac{\d{y}}{\d{t}} = -a(r_+-r_-)y+ay^2 = ay[y-(r_+-r_-)], \qquad
y(0) =\|u_{\mathrm{in}}\|-r_-.
\end{equation}
Since the homogeneous equation has the closed-form solution
\begin{equation}
y(t) = \frac{r_+-r_-}{1-e^{a(r_+-r_-)t}[1-(r_+-r_-)/(\|u_{\mathrm{in}}\|-r_-)]},
\end{equation}
the solution of the inhomogeneous equation can be obtained as
\begin{equation}
x(t) = \frac{r_+-r_-}{1-e^{a(r_+-r_-)t}[1-(r_+-r_-)/(\|u_{\mathrm{in}}\|-r_-)]} + r_-.
\end{equation}
Therefore we have
\begin{equation}
\|u(t)\| \le \frac{r_+-r_-}{1-e^{a(r_+-r_-)t}[1-(r_+-r_-)/(\|u_{\mathrm{in}}\|-r_-)]} + r_-.
\label{eq:solution-norm}
\end{equation}
Since $\Rq < 1 \Leftrightarrow a\|u_{\mathrm{in}}\|+\frac{c}{\|u_{\mathrm{in}}\|}<-b \Leftrightarrow a\|u_{\mathrm{in}}\|^2+b\|u_{\mathrm{in}}\|+c<0$, $\|u_{\mathrm{in}}\|$ is located between the two roots $r_-$ and $r_+$, and thus $1-(r_+-r_-)/(\|u_{\mathrm{in}}\|-r_-)<0$. This implies $\|u(t)\|$ in \eq{solution-norm} decreases from $u(0) = \|u_{\mathrm{in}}\|$, so we have $\|u(t)\|<\|u_{\mathrm{in}}\|<r_+$ for any $t>0$.
\end{proof}

We remark that $\lim_{t\to\infty}\frac{\d\|u\|}{\d{t}}=0$ since $\frac{\d\|u\|}{\d{t}}<0$ and $\|u(t)\|\ge0$, so $u(t)$ approaches to a stationary point of the right-hand side of \eq{NODE} (called an attractor in the theory of dynamical systems). 

Note that for a homogeneous equation (i.e., $\|F_0\|=0$), this shows that the dissipation inevitably leads to exponential decay. In this case we have $r_-=0$, so \eq{solution-norm} gives
\begin{align}
  \|u(t)\| 
  = \frac{\|u_{\mathrm{in}}\| r_+}{e^{a r_+ t}(r_+ - \|u_{\mathrm{in}}\|) + \|u_{\mathrm{in}}\|},
  \label{eq:homo_example}
\end{align}
which decays exponentially in $t$.

On the other hand, as mentioned in the introduction, the solution of a dissipative inhomogeneous quadratic ODE can remain asymptotically nonzero. Here we present an example of this. Consider a time-independent uncoupled system with $\frac{\d{u_j}}{\d{t}} = f_2u_j^2+f_1u_j+f_0$, $j\in\range{n}$, with $u_j(0)=x_0>0$, $f_2>0$, $f_1<0$, $f_0>0$, and $\Rq < 1$. We see that each $u_j(t)$ decreases from $x_0$ to $x_1 \coloneqq \frac{-f_1-\sqrt{f_1^2-4f_2f_0}}{2f_2} >0$, with $0<x_1<u_j(t)<x_0$. Hence, the norm of $u(t)$ is bounded as $0<\sqrt{n}x_1<\|u(t)\|<\sqrt{n}x_0$ for any $t>0$. In general, it is hard to lower bound $\|u(t)\|$, but the above example shows that a nonzero inhomogeneity can prevent the norm of the solution from decreasing to zero.

We now give an upper bound on the error of Carleman linearization. 

\begin{lemma}\label{lem:Carleman}
Consider an instance of the quadratic ODE \eq{NODE}, with its corresponding Carleman linearization as defined in \eq{LODE}. As in \prb{node}, assume that the eigenvalues $\lambda_j$ of $F_1$ satisfy $\Re{(\lambda_n)} \le \cdots \le \Re{(\lambda_1)} < 0$.
Assume that $\Rq$ defined in \eq{A1} satisfies $\Rq<1$ and assume that $\|F_0\| \le \|F_2\|$.
Then for any $j \in \range{N}$, the error $\eta_j(t) \coloneqq u^{\otimes j}(t)-\hat y_j(t)$ satisfies
\begin{equation}
\|\eta_j(t)\| \le \|\eta(t)\| \le tN\|F_2\|\|u_{\mathrm{in}}\|^{N+1}.
\label{eq:bound_general_1}
\end{equation}
\end{lemma}

\begin{proof}
The exact solution $u(t)$ of the original quadratic ODE \eq{NODE} satisfies
\begin{equation}
\frac{\d{}}{\d{t}}
  \begin{pmatrix}
    u \\
    u^{\otimes 2} \\
    u^{\otimes 3} \\
    \vdots \\
    u^{\otimes(N-1)} \\
    u^{\otimes N} \\
    \vdots \\
  \end{pmatrix}
=
  \begin{pmatrix}
    A_1^1 & A_2^1 &  &  &  &  &  \\
    A_1^2 & A_2^2 & A_3^2 & &  &  &  \\
     & A_2^3 & A_3^3 & A_4^3 &  &  &  \\
     &  & \ddots & \ddots & \ddots &  &  \\
     &  &  & A_{N-2}^{N-1} & A_{N-1}^{N-1} & A_N^{N-1} &  \\
     &  &  &  & A_{N-1}^N & A_N^N & \ddots \\
     &  &  &  &  & \ddots & \ddots \\
  \end{pmatrix}
  \begin{pmatrix}
    u \\
    u^{\otimes 2} \\
    u^{\otimes 3} \\
    \vdots \\
    u^{\otimes(N-1)} \\
    u^{\otimes N} \\
    \vdots \\
  \end{pmatrix}+
  \begin{pmatrix}
    F_0(t) \\
    0 \\
    0 \\
    \vdots \\
    0 \\
    0 \\
    \vdots \\
  \end{pmatrix},
\label{eq:LUODE}
\end{equation}
and the approximated solution $\hat y_j(t)$ satisfies \eq{UODE}. Comparing these equations, we have
\begin{equation}
  \frac{\d{\eta}}{\d{t}} = A(t) \eta + \hat b(t), \qquad
  \eta(0) = 0
\label{eq:ELODE}
\end{equation}
with the tri-diagonal block structure
\begin{equation}
\frac{\d{}}{\d{t}}
  \begin{pmatrix}
    \eta_1 \\
    \eta_2 \\
    \eta_3 \\
    \vdots \\
    \eta_{N-1} \\
    \eta_N \\
  \end{pmatrix}
=
  \begin{pmatrix}
    A_1^1 & A_2^1 &  &  &  &  \\
    A_1^2 & A_2^2 & A_3^2 & &  &  \\
     & A_2^3 & A_3^3 & A_4^3 &  &  \\
     &  & \ddots & \ddots & \ddots &  \\
     &  &  & A_{N-2}^{N-1} & A_{N-1}^{N-1} & A_N^{N-1} \\
     &  &  &  & A_{N-1}^N & A_N^N \\
  \end{pmatrix}
  \begin{pmatrix}
    \eta_1 \\
    \eta_2 \\
    \eta_3 \\
    \vdots \\
    \eta_{N-1} \\
    \eta_N \\
  \end{pmatrix}+
  \begin{pmatrix}
    0 \\
    0 \\
    0 \\
    \vdots \\
    0 \\
    A_{N+1}^Nu^{\otimes(N+1)} \\
  \end{pmatrix},
\label{eq:EUODE}
\end{equation}

Consider the derivative of $\|\eta(t)\|$. We have
\begin{equation}
\frac{\d\|\eta\|^2}{\d{t}} = \eta^\dag(A(t) + A^\dag(t))\eta + \eta^\dag \hat b(t) + \hat b(t)^\dag \eta.
\label{eq:etabound}
\end{equation}
For $\eta^\dag(A(t) + A^\dag(t))\eta$, we bound each term as
\begin{equation}
\begin{aligned}
\eta_j^\dag A_{j+1}^j\eta_{j+1}+\eta_{j+1}^\dag (A_{j+1}^j)^\dag\eta_j &\le j\|F_2\|\|\eta_{j+1}\|\|\eta_j\|, \\
\eta_j^\dag [A_j^j+(A_j^j)^\dag]\eta_j &\le 2j\Re(\lambda_1)\|\eta_j\|^2, \\
\eta_j^\dag A_{j-1}^j\eta_{j-1}+\eta_{j-1}^\dag (A_{j-1}^j)^\dag\eta_j &\le 2j\|F_0\|\|\eta_{j-1}\|\|\eta_j\|
\end{aligned}
\end{equation}
using the definitions in \eq{tensor2}--\eq{tensor0}.

Define a matrix $G\in\R^{N\times N}$ with nonzero entries $G_{j-1,j} = j\|F_0\|$, $G_{j,j} = j\Re(\lambda_1)$, and $G_{j+1,j} = j\|F_2\|$. Then
\begin{equation}
\eta^\dag(A(t) + A^\dag(t))\eta \le \tilde{\eta}^\dag(G + G^\dag)\tilde{\eta}
\label{eq:eta_a}
\end{equation}
where $\tilde\eta \in \R^N$ with $\tilde\eta_j = \|\eta_j\|$. Since $\|F_2\| + \|F_0\| < |{\Re{(\lambda_1)}}|$ and $\|F_0\| \le \|F_2\|$, $G+ G^\dag$ is strictly diagonally dominant and thus the eigenvalues $\nu_j$ of $(G+G^\dag)/2$ satisfy $\nu_n \le \cdots \le \nu_1 < 0$. Thus we have
\begin{equation}
\tilde{\eta}^\dag(G + G^\dag)\tilde{\eta} \le 2\nu_1\|\eta\|^2.
\label{eq:eta_g}
\end{equation}
For $\eta^\dag \hat b(t) + \hat b(t)^\dag \eta$, we have
\begin{equation}
\eta^\dag \hat b(t) + \hat b(t)^\dag \eta \le 2\|\hat b\|\|\eta\| \le 2\|A_{N+1}^N\|\|u^{\otimes(N+1)}\|\|\eta\|.
\end{equation}
Since $\|A_{N+1}^N\| = N\|F_2\|$, and $\|u^{\otimes(N+1)}\|=\|u\|^{N+1} \le \|u_{\mathrm{in}}\|^{N+1}$, we have
\begin{equation}
\eta^\dag \hat b(t) + \hat b(t)^\dag \eta \le 2N\|F_2\|\|u_{\mathrm{in}}\|^{N+1}\|\eta\|.
\label{eq:eta_b}
\end{equation}

Using \eq{eta_a}, \eq{eta_g}, and \eq{eta_b} in \eq{etabound}, we find
\begin{equation}
\frac{\d\|\eta\|^2}{\d{t}} \le 2\nu_1\|\eta\|^2 + 2N\|F_2\|\|u_{\mathrm{in}}\|^{N+1}\|\eta\|,
\end{equation}
so elementary calculus gives
\begin{equation}
\frac{\d\|\eta\|}{\d{t}} 
= \frac{1}{2\|\eta\|}\frac{\d\|\eta\|^2}{\d{t}}
\le \nu_1\|\eta\| + N\|F_2\|\|u_{\mathrm{in}}\|^{N+1}.
\end{equation}
Solving the differential inequality as an equation with $\eta(0) = 0$ gives us a bound on $\|\eta\|$:
\begin{equation}
\|\eta(t)\| \le \int_0^te^{\Re(\nu_1)(t-s)}N\|F_2\|\|u_{\mathrm{in}}\|^{N+1}\,\d{s} \le N\|F_2\|\|u_{\mathrm{in}}\|^{N+1}\int_0^te^{\nu_1(t-s)}\d{s}.
\label{eq:eta_norm}
\end{equation}
Finally, using
\begin{equation}
\int_0^te^{\nu_1(t-s)}\d{s}=\frac{1-e^{\nu_1t}}{|\nu_1|}\le t
\label{eq:norm_negative_t}
\end{equation}
(where we used the inequality $1-e^{at}\le -at$ with $a<0$), \eq{eta_norm} gives the bound
\begin{equation}
\|\eta_j(t)\| \le \|\eta(t)\| \le tN\|F_2\|\|u_{\mathrm{in}}\|^{N+1}
\label{eq:eta_norm_1}
\end{equation}
as claimed. 
\end{proof}

Note that \eq{norm_negative_t} can be bounded alternatively by
\begin{equation}
\int_0^te^{\nu_1(t-s)}\d{s}=\frac{1-e^{\nu_1t}}{|\nu_1|}\le\frac{1}{|\nu_1|},
\label{eq:norm_negative_nu}
\end{equation}
and thus $\|\eta_j(t)\| \le \|\eta(t)\| \le \frac{1}{|\nu_1|}N\|F_2\|\|u_{\mathrm{in}}\|^{N+1}$. We select \eq{norm_negative_t} because it avoids including an additional parameter $\nu_1$.

We also give an improved analysis that works for homogeneous quadratic ODEs ($F_0(t)=0$) under milder conditions. This analysis follows the proof in \cite{FP17} closely.

\begin{corollary}\label{cor:Carleman_homo}
Under the same setting of \lem{Carleman}, assume $F_0(t)=0$ in \eq{NODE}.
Then for any $j \in \range{N}$, the error $\eta_j(t) \coloneqq u^{\otimes j}(t)-\hat y_j(t)$ satisfies
\begin{equation}
\|\eta_j(t)\| \le \|u_{\mathrm{in}}\|^j\Rq^{N+1-j}.
\label{eq:bound_general}
\end{equation}
For $j=1$, we have the tighter bound
\begin{equation}
\|\eta_1(t)\| \le
\|u_{\mathrm{in}}\| \Rq^N\big(1-e^{\Re(\lambda_1)t}\big)^{N}.
\label{eq:bound_1}
\end{equation}
\end{corollary}

\begin{proof}
We again consider $\eta$ satisfying \eq{ELODE}. Since $F_0(t)=0$, \eq{ELODE} reduces to a time-independent ODE with an upper triangular block structure,
\begin{equation}
\frac{\d{\eta_j}}{\d{t}} = A_j^j\eta_j+A_{j+1}^j\eta_{j+1}, \quad j\in\range{N-1}
\label{eq:deta}
\end{equation}
and
\begin{equation}
\frac{\d{\eta_N}}{\d{t}} = A_N^N\eta_N+A_{N+1}^Nu^{\otimes(N+1)}.
\label{eq:EODE_N}
\end{equation}

We proceed by backward substitution. Since $\eta_N(0)=0$, we have
\begin{equation}
\eta_N(t) = \int_0^te^{A_N^N(t-s_0)}A_{N+1}^Nu^{\otimes(N+1)}(s_0)\,\d{s_0}.
\label{eq:eta_N}
\end{equation}
For $j \in \range{N}$, \eq{tensor1} gives
$\|e^{A_j^j t}\| = e^{j\Re(\lambda_1)t}$
and \eq{tensor2} gives $\|A_{j+1}^j\| = j \norm{F_2}$.
By \lem{decrease}, $\|u^{\otimes(N+1)}\|=\|u\|^{N+1} \le \|u_{\mathrm{in}}\|^{N+1}$. We can therefore upper bound \eq{eta_N} by
\begin{equation}
\begin{aligned}
\|\eta_N(t)\|
&\le \int_0^t\|e^{A_N^N(t-s_0)}\|\cdot\|A_{N+1}^Nu^{\otimes(N+1)}(s_0)\|\,\d{s_0} \\
&\le N\|F_2\|\|u_{\mathrm{in}}\|^{N+1}\int_0^te^{N\Re(\lambda_1)(t-s_0)}\,\d{s_0}.
\label{eq:norm_N}
\end{aligned}
\end{equation}
For $j=N-1$, \eq{deta} gives
\begin{equation}
\frac{\d{\eta_{N-1}}}{\d{t}} = A_{N-1}^{N-1}\eta_{N-1}+A_N^{N-1}\eta_N.
\label{eq:EODE_N_1}
\end{equation}
Again, since $\eta_{N-1}(0)=0$, we have
\begin{equation}
\eta_{N-1}(t) = \int_0^te^{A_{N-1}^{N-1}(t-s_1)}A_N^{N-1}\eta_N(s_1)\,\d{s_1}
\label{eq:eta_N_1}
\end{equation}
which has the upper bound
\begin{equation}
\begin{aligned}
\|\eta_{N-1}(t)\|
&\le \int_0^t\|e^{A_{N-1}^{N-1}(t-s_1)}\|\cdot\|A_N^{N-1}\eta_N(s_1)\|\,\d{s_1} \\
&\le (N-1)\|F_2\|\int_0^te^{(N-1)\Re(\lambda_1)(t-s_1)}\|\eta_N(s_1)\|\,\d{s_1} \\
&\le N(N-1)\|F_2\|^2\|u_{\mathrm{in}}\|^{N+1}\int_0^te^{(N-1)\Re(\lambda_1)(t-s_1)}\int_0^{s_1}e^{N\Re(\lambda_1)(s_1-s)}\,\d{s_0}\,\d{s_1},
\label{eq:norm_N_1}
\end{aligned}
\end{equation}
where we used \eq{norm_N} in the last step.
Iterating this procedure for $j=N-2, \ldots, 1$, we find
\begin{equation}
\begin{aligned}
\|\eta_{j}(t)\|
&\le \frac{N!}{(j-1)!}\|F_2\|^{N+1-j}\|u_{\mathrm{in}}\|^{N+1} \int_0^te^{j\Re(\lambda_1)(t-s_{N-j})}\int_0^{s_{N-j}}e^{(j+1)\Re(\lambda_1)(s_{N-j}-s_{N-1-j})}\cdots \\
&\qquad \int_0^{s_2}e^{(N-1)\Re(\lambda_1)(s_2-s_1)}\int_0^{s_1}e^{N\Re(\lambda_1)(s_1-s_0)}\,\d{s_0}\,\d{s_1}\cdots \d{s_{N-j-1}} \, \d{s_{N-j}} \\
&= \frac{N!}{(j-1)!}\|F_2\|^{N+1-j}\|u_{\mathrm{in}}\|^{N+1} \!\int_0^{s_{N+1-j}}\!\!\cdots\!\int_0^{s_2}\!\!\int_0^{s_1} \!\!e^{\Re(\lambda_1)(-Ns_0+\sum_{k=1}^{N-j+1}s_k)} \, \d{s_0}\,\d{s_1}\cdots\d{s_{N-j}}.
\label{eq:norm_j}
\end{aligned}
\end{equation}

Finally, using
\begin{equation}
\int_0^{s_{k+1}}e^{(N-k)\Re(\lambda_1)(s_{k+1}-s_k)}\d{s_k}=\frac{1-e^{(N-k)\Re(\lambda_1)s_{k+1}}}{(N-k)|{\Re(\lambda_1)}|}\le\frac{1}{(N-k)|{\Re(\lambda_1)}|}
\label{eq:norm_negative}
\end{equation}
for $k=0,\ldots,N-j$, \eq{norm_j} can be bounded by
\begin{equation}
\|\eta_{j}(t)\| \le \frac{N!}{(j-1)!}\|F_2\|^{N+1-j}\|u_{\mathrm{in}}\|^{N+1}\frac{(j-1)!}{N!|{\Re(\lambda_1)}|^{N+1-j}} = \frac{\|u_{\mathrm{in}}\|^{N+1}\|F_2\|^{N+1-j}}{|{\Re(\lambda_1)}|^{N+1-j}} =  \|u_{\mathrm{in}}\|^j\Rq^{N+1-j}.
\label{eq:norm_general}
\end{equation}

For $j=1$, the bound can be further improved. By Lemma 5.2 of \cite{FP17}, if $a\ne 0$,
\begin{equation}
\int_0^{s_{N}}\!\cdots\int_0^{s_2}\int_0^{s_1}e^{a(-Ns_0+\sum_{k=1}^Ns_k)}\,\d{s_0}\,\d{s_1}\cdots\d{s_{N-1}} = \frac{(e^{as_N}-1)^N}{N!a^N}.
\label{eq:integration}
\end{equation}
With $s_N=t$ and $a=\Re(\lambda_1)$, we find
\begin{equation}
\begin{aligned}
\|\eta_{1}(t)\| 
&\le N!\|F_2\|^N\|u_{\mathrm{in}}\|^{N+1} \int_0^{s_N}\!\cdots\int_0^{s_2}\int_0^{s_1} e^{\Re(\lambda_1)(-Ns_0+\sum_{k=1}^Ns_k)} \, \d{s_0}\,\d{s_1}\cdots\d{s_{N-1}} \\
&\le N!\|F_2\|^N\|u_{\mathrm{in}}\|^{N+1}\frac{(e^{\Re(\lambda_1)t}-1)^N}{N!\Re(\lambda_1)^N} \\
&=  \|u_{\mathrm{in}}\|\Rq^N\big(1-e^{at}\big)^N,
\label{eq:norm_1}
\end{aligned}
\end{equation}
which is tighter than the $j=1$ case in \eq{norm_general}.
\end{proof}

While \prb{node} makes some strong assumptions about the system of differential equations, they appear to be necessary for our analysis. Specifically, the conditions $\Re{(\lambda_1)} < 0$ and $\Rq<1$ are required to ensure arbitrary-time convergence.

Since the Euler method for \eq{LODE} is unstable if $\Re{(\lambda_1)}>0$ \cite{Ber14,Dah63}, we only consider the case $\Re{(\lambda_1)}\le0$. If $\Re{(\lambda_1)}=0$, \eq{integration} reduces to
\begin{equation}
\int_0^{s_{N}}\!\cdots\int_0^{s_2}\int_0^{s_1}e^{a(-Ns_0+\sum_{k=1}^Ns_k)}\,\d{s_0}\,\d{s_1}\cdots\d{s_{N-1}} = \frac{t^N}{N!},
\end{equation}
giving the error bound
\begin{equation}
\|\eta_1(t)\| \le \|u_{\mathrm{in}}\|(\|u_{\mathrm{in}}\|\|F_2\|t)^N
\end{equation}
instead of \eq{norm_1}. Then the error bound can be made arbitrarily small for a finite time by increasing $N$, but after $t > {1}/{\|u_{\mathrm{in}}\| \|F_2\|}$, the error bound diverges.

Furthermore, if $\Rq \ge 1$, Bernoulli's inequality gives
\begin{equation}
\|u_{\mathrm{in}}\|(1-Ne^{\Re(\lambda_1)t}) \le \|u_{\mathrm{in}}\|(1-Ne^{\Re(\lambda_1)t})\Rq^N  \le \|u_{\mathrm{in}}\|(1-e^{\Re(\lambda_1)t})^N\Rq^N,
\end{equation}
where the right-hand side upper bounds $\|\eta_1(t)\|$ as in \eq{norm_1}. Assuming $\norm{\eta_1(t)} = \norm{u(t)-\hat y_1(t)}$ is smaller than $\|u_{\mathrm{in}}\|$, we require
\begin{equation}
N = \Omega (e^{|\Re(\lambda_1)|t}).
\end{equation}
In other words, to apply \eq{bound_1} for the Carleman linearization procedure, the truncation order given by \lem{Carleman} must grow exponentially with $t$.

In fact, we prove in \sec{lower} that for $\Rq\ge\sqrt{2}$, \emph{no} quantum algorithm (even one based on a technique other than Carleman linearization) can solve \prb{node} efficiently. It remains open to understand the complexity of the problem for $1 \le \Rq < \sqrt2$.

On the other hand, if $\Rq<1$, both \eq{bound_general} and \eq{bound_1} decrease exponentially with $N$, making the truncation efficient. We discuss the specific choice of $N$ in \eq{N} below.

\subsubsection{Error from forward Euler method}

Next, we provide an upper bound for the error incurred by approximating \eq{LODE} with the forward Euler method. This problem has been well studied for general ODEs. Given an ODE $\frac{\d{z}}{\d{t}}=f(z)$ on $[0,T]$ with an inhomogenity $f$ that is an $L$-Lipschitz continuous function of $z$, the global error of the solution is upper bounded by $e^{LT}$, although in most cases this bound overestimates the actual error \cite{Atk08}. To remove the exponential dependence on $T$ in our case, we derive a tighter bound for time discretization of \eq{LODE} in \lem{global} below. This lemma is potentially useful for other ODEs as well and can be straightforwardly adapted to other problems.

\begin{lemma}\label{lem:global}
Consider an instance of the quantum quadratic ODE problem as defined in \prb{node}, with $\Rq<1$ as defined in \eq{A1}. Choose a time step
\begin{equation}
h \le \min \biggl\{ \frac{1}{N\|F_1\|}, \frac{2(|\Re(\lambda_1)|-\|F_2\|-\|F_0\|)}{N(|\Re(\lambda_1)|^2 - (\|F_2\|+\|F_0\|)^2 + \|F_1\|^2)} \biggr\}
\label{eq:h}
\end{equation}
in general, or
\begin{equation}
h \le \frac{1}{N\|F_1\|}
\label{eq:h_real}
\end{equation}
if the eigenvalues of $F_1$ are all real. Suppose the error from Carleman linearization $\eta(t)$ as defined in \lem{Carleman} is bounded by
\begin{equation}
\|\eta(t)\| \le \frac{g}{4},
\label{eq:g_condition_1}
\end{equation}
where $g$ is defined in \eq{gq}. Then the global error from the forward Euler method \eq{forward} on the interval $[0,T]$ for \eq{LODE} satisfies
\begin{equation}
\|\hat y_1(T)-y^m_1\| 
\le
\|\hat y(T)-y^m\| 
\le 3N^{2.5}Th[(\|F_2\|+\|F_1\|+\|F_0\|)^2+\|F_0'\|].
\end{equation}
\end{lemma}

\begin{proof}
We define a linear system that locally approximates the initial value problem \eq{LODE} on $[kh, (k+1)h]$ for $k\in\rangez{m}$ as
\begin{equation}
z^k = [I+A((k-1)h)h] \hat y((k-1)h) + b((k-1)h),
\label{eq:local_forward}
\end{equation}
where $\hat y(t)$ is the exact solution of \eq{LODE}. For $k \in \range{m}$, we denote the local truncation error by 
\begin{equation}
e^k \coloneqq \|\hat y(kh)-z^k\|
\label{eq:local_error}
\end{equation}
and the global error by
\begin{equation}
g^k \coloneqq \|\hat y(kh)-y^k\|,
\label{eq:global_error}
\end{equation}
where $y^k$ in \eq{forward} is the numerical solution. Note that $g^m = \|\hat y(T)-y^m\|$.

For the local truncation error, we Taylor expand $\hat y((k-1)h)$ with a second-order Lagrange remainder, giving
\begin{equation}
\hat y(kh) = \hat y((k-1)h) + \hat y'((k-1)h)h + \frac{\hat y''(\xi)h^2}{2}
\end{equation}
for some $\xi\in[(k-1)h,kh]$. Since $\hat y'((k-1)h) = A((k-1)h)\hat y((k-1)h)+b((k-1)h)$ by \eq{LODE}, we have
\begin{equation}
\hat y(kh) = [I+A((k-1)h)h]\hat y((k-1)h)+b((k-1)h) + \frac{\hat y''(\xi)h^2}{2} = z^k + \frac{\hat y''(\xi)h^2}{2},
\end{equation}
and thus the local error \eq{local_error} can be bounded as
\begin{equation}
e^k = \|\hat y(kh)-z^k\| = \|\hat y''(\xi)h\|\frac{h^2}{2} \le \frac{Mh^2}{2},
\label{eq:local}
\end{equation}
where $M \coloneqq \max_{t\in[0,T]}\|\hat y''(\tau)\|$.

By the triangle inequality, the global error \eq{global_error} can therefore be bounded as
\begin{equation}
g^k = \|\hat y(kh)-y^k\| \le \|\hat y(kh)-z^k\|+\|z^k-y^k\| \le e^k + \|z^k-y^k\|.
\label{eq:global_local}
\end{equation}
Since $y^k$ and $z^k$ are obtained by the same linear system with different right-hand sides, we have the upper bound
\begin{equation}
\begin{aligned}
\|z^k-y^k\| &= \|[I+A((k-1)h)h][\hat y((k-1)h)-y^{k-1}]\| \\
&\le \|I+A((k-1)h)h\| \cdot \|\hat y((k-1)h)-y^{k-1}\| \\
&= \|I+A((k-1)h)h\|g^{k-1}.
\label{eq:zkyk}
\end{aligned}
\end{equation}

In order to provide an upper bound for  $\|I+A(t)h\|$ for all $t\in[0,T]$, we write
\begin{equation}
I+A(t)h=H_2+H_1+H_0(t)
\end{equation}
where
\begin{align}
H_2 &= \sum_{j=1}^{N-1}|j\rangle\langle j+1|\otimes A_{j+1}^jh, \\
H_1 &= I+\sum_{j=1}^N|j\rangle\langle j|\otimes A_j^jh, \\
H_0(t) &= \sum_{j=2}^{N}|j\rangle\langle j-1|\otimes A_{j-1}^jh.
\end{align}
We provide upper bounds separately for $\|H_2\|$, $\|H_1\|$, and $\|H_0\| \coloneqq \max_{t\in[0,T]} \|H_0(t)\|$, and use the bound $\max_{t\in[0,T]}\|I+A(t)h\| \le \|H_2\| + \|H_1\| + \|H_0\|$.

The eigenvalues of $A_j^j$ consist of all $j$-term sums of the eigenvalues of $F_1$. More precisely, they are $\{\sum_{\ell \in \range{j}} \lambda_{\mathcal{I}^j_\ell}\}_{\mathcal{I}^j \in \range{n}^j}$ where $\{\lambda_\ell\}_{\ell \in \range{n}}$ are the eigenvalues of $F_1$ and $\mathcal{I}^j \in \range{n}^j$ is a $j$-tuple of indices.
The eigenvalues of $H_1$ are thus $\{1 + h\sum_{\ell \in \range{j}} \lambda_{\mathcal{I}_\ell^j}\}_{\mathcal{I}^j \in \range{n}^j,j \in \range{N}}$.
With $J \coloneqq \max_{\ell \in \range{n}} |\Im(\lambda_\ell)|$, we have
\begin{equation}
  \begin{aligned}
\Big|1+h \sum_{\ell \in \range{j}}\lambda_{\mathcal{I}^j_\ell}\Big|^2 &= \Big|1+h \sum_{\ell \in \range{j}} \Re(\lambda_{\mathcal{I}^j_\ell}) \Big|^2 + \Big|h\sum_{\ell \in \range{j}} \Im(\lambda_{\mathcal{I}^j_\ell})\Big|^2 \\
&\le 1 - 2Nh|\Re(\lambda_1)| + N^2h^2(|\Re(\lambda_1)|^2+J^2)\\
  \end{aligned}
\end{equation}
where we used $Nh|\Re(\lambda_1)| \le 1$ by \eq{h}. Therefore
\begin{equation}
\|H_1\| 
= \max_{j \in \range{N}} \max_{\mathcal{I}^j \in \range{n}^j} \Big|1+h \sum_{\ell \in \range{j}}\lambda_{\mathcal{I}^j_\ell}\Big| 
\le \sqrt{1 - 2Nh|\Re(\lambda_1)| + N^2h^2(|\Re(\lambda_1)|^2 + J^2)}.
\label{eq:norm_H_1}
\end{equation}

We also have
\begin{equation}
\|H_2\| = \biggl\|\sum_{j=1}^{N-1}|j\rangle\langle j+1|\otimes A_{j+1}^jh\biggr\| \le \max_{j \in \range{N}} \|A_{j+1}^j\|h \le N\|F_2\|h
\label{eq:norm_H_2}
\end{equation}
and
\begin{equation}
\|H_0\| = \max_{t\in[0,T]}\biggl\|\sum_{j=1}^{N-1}|j\rangle\langle j01|\otimes A_{j-1}^jh\biggr\| \le \max_{t\in[0,T]}\max_{j \in \range{N}} \|A_{j-1}^j\|h \le N\max_{t\in[0,T]}\|F_0(t)\|h \le N\|F_0\|h.
\label{eq:norm_H_0}
\end{equation}

Using the bounds \eq{norm_H_1} and \eq{norm_H_2}, we aim to select the value of $h$ to ensure
\begin{equation}
\max_{t\in[0,T]}\|I+A(t)h\| \le \|H_2\| + \|H_1\| + \|H_0\| \le 1.
\label{eq:norm_equivalent}
\end{equation}
The assumption \eq{h} implies
\begin{equation}
h \le \frac{2(|\Re(\lambda_1)|-\|F_2\|-\|F_0\|)}{N(|\Re(\lambda_1)|^2 - (\|F_2\|+\|F_0\|)^2 + J^2)}
\label{eq:h_general}
\end{equation}
(note that the denominator is non-zero due to \eq{A2}).
Then we have
\begin{equation}
  \begin{aligned}
&& N^2h^2(|\Re(\lambda_1)|^2 - \|F_2\|^2 + J^2) &\le 2Nh(|\Re(\lambda_1)|-\|F_2\|-\|F_0\|) \\
&\implies &
1 - 2Nh|\Re(\lambda_1)| + N^2h^2(|\Re(\lambda_1)|^2 + J^2) &\le 1 - 2N(\|F_2\|+\|F_0\|)h + N^2(\|F_2\|+\|F_0\|)^2h^2 \\
&\implies &
\sqrt{1 - 2Nh|\Re(\lambda_1)| + N^2h^2(|\Re(\lambda_1)|^2 + J^2)} &\le 1 - N(\|F_2\|+\|F_0\|)h, \\
  \end{aligned}
\end{equation}
so $\max_{t\in[0,T]}\|I + A(t)h\| \le 1$ as claimed.

The choice \eq{h} can be improved if an upper bound on $J$ is known. In particular, if $J=0$, \eq{h_general} simplifies to
\begin{equation}
h \le \frac{2}{N(|\lambda_1| + \|F_2\| + \|F_0\|)},
\end{equation}
which is satisfied by \eq{h_real} using $\|F_2\| + \|F_0\|<|\lambda_1|\le\|F_1\|$.

Using this in \eq{zkyk}, we have
\begin{equation}
\|z^k-y^k\| \le \|I+A((k-1)h)h\|g^{k-1} \le g^{k-1}.
\label{eq:global_increase}
\end{equation}
Plugging \eq{global_increase} into \eq{global_local} iteratively, we find
\begin{equation}
g^k \le g^{k-1}+e^k \le g^{k-2}+e^{k-1}+e^k \le \cdots \le \sum_{j=1}^ke^j, \quad k\in\rangez{m+1}.
\label{eq:global_recurrence}
\end{equation}
Using \eq{local}, this shows that the global error from the forward Euler method is bounded by
\begin{equation}
\|\hat y^1(kh)-y^k_1\| \le \|\hat y(kh)-y^k\| = g^k \le \sum_{j=1}^ke^j \le \frac{Mkh^2}{2},
\label{eq:half_global_0}
\end{equation}
and when $k=m$, $mh=T$,
\begin{equation}
\|\hat y^1(T)-y^m_1\| \le g^m \le m\frac{Mh^2}{2} = \frac{MTh}{2}.
\label{eq:global_0}
\end{equation}

Finally, we remove the dependence on $M= \max_{t\in[0,T]}\|\hat y''(\tau)\|$. Since
\begin{equation}
\begin{aligned}
\hat y''(t) &= [A(t)\hat y(t)+b(t)]' = A(t)\hat y'(t)+ A'(t)\hat y(t) + b'(t) \\
&= A(t)[A(t)\hat y(t)+b(t)] + A'(t)\hat y(t) + b'(t),
\end{aligned}
\end{equation}
we have
\begin{equation}
\|\hat y''(t)\| = \|A(t)\|^2\|\hat y(t)\| + \|A(t)\|\|b(t)\| + \|A'(t)\|\|\hat y(t)\| + \|b'(t)\|.
\label{eq:bound_y2}
\end{equation}
Maximizing each term for $t\in[0,T]$, we have
\begin{equation}
\max_{t\in[0,T]}\|A(t)\| = \biggl\|\sum_{j=1}^{N-1}|j\rangle\langle j+1|\otimes A_{j+1}^j+\sum_{j=1}^N|j\rangle\langle j|\otimes A_j^j+\sum_{j=2}^{N}|j\rangle\langle j-1|\otimes A_{j-1}^j\biggr\| \le N(\|F_2\|+\|F_1\|+\|F_0\|),
\label{eq:bound_A0}
\end{equation}
\begin{equation}
\max_{t\in[0,T]}\|A'(t)\| = \biggl\|\sum_{j=2}^{N}|j\rangle\langle j-1|\otimes (A_{j-1}^j)'\biggr\| \le N\|F'_0(t)\| \le N\|F'_0\|,
\label{eq:bound_A1}
\end{equation}
\begin{equation}
\max_{t\in[0,T]}\|b(t)\| = \max_{t\in[0,T]}\|F_0(t)\| \le \|F_0\|,
\label{eq:bound_b0}
\end{equation}
\begin{equation}
\max_{t\in[0,T]}\|b'(t)\| = \max_{t\in[0,T]}\|F_0'(t)\| \le \|F_0'\|,
\label{eq:bound_b1}
\end{equation}
and using $\|u\| \le \|u_{\mathrm{in}}\| < 1$, $\|\eta_j(t)\| \le \|\eta(t)\| \le g/4 < \|u_{\mathrm{in}}\|/4$ by \eq{g_condition_1}, and $\Rq < 1$, we have
\begin{equation}
\begin{aligned}
\|\hat y(t)\|^2
&\le \sum_{j=1}^N\|\hat y_j(t)\|^2
= \sum_{j=1}^N\|u^{\otimes j}(t) - \eta_j(t)\|^2
\le 2\sum_{j=1}^N(\|u^{\otimes j}(t)\|^2 + \|\eta_j(t)\|^2) \\
&\le 2\sum_{j=1}^N\biggl(\|u_{\mathrm{in}}\|^{2j} + \frac{\|u_{\mathrm{in}}\|^2}{16}\biggr)
< 2\sum_{j=1}^N\biggl(1 + \frac{1}{16} \biggr)\|u_{\mathrm{in}}\|^2
< 4N\|u_{\mathrm{in}}\|^2
< 4N
\label{eq:parallel_inequality}
\end{aligned}
\end{equation}
for all $t \in [0,T]$. Therefore, substituting the bounds \eq{bound_A0}--\eq{parallel_inequality} into \eq{bound_y2}, we find
\begin{equation}
\begin{aligned}
M &\le \max_{t\in[0,T]} \|A(t)\|^2\|\hat y(t)\| + \|A(t)\|\|b(t)\| + \|A'(t)\|\|\hat y(t)\| + \|b'(t)\| \\
&\le 2N^{2.5}(\|F_2\|+\|F_1\|+\|F_0\|)^2 + N(\|F_2\|+\|F_1\|+\|F_0\|)\|F_0\| + 2N^{1.5}\|F'_0\| + \|F_0'\| \\
&\le 6N^{2.5}[(\|F_2\|+\|F_1\|+\|F_0\|)^2+\|F_0'\|].
\end{aligned}
\end{equation}
Thus, \eq{global_0} gives
\begin{equation}
\|\hat y^1(kh)-y^m_1\| \le 3N^{2.5}kh^2[(\|F_2\|+\|F_1\|+\|F_0\|)^2+\|F_0'\|],
\label{eq:half_global}
\end{equation}
and when $k=m$, $mh=T$,
\begin{equation}
\|\hat y^1(T)-y^m_1\| \le 3N^{2.5}Th[(\|F_2\|+\|F_1\|+\|F_0\|)^2+\|F_0'\|]
\label{eq:global}
\end{equation}
as claimed.
\end{proof}

\subsection{Condition number}

Now we upper bound the condition number of the linear system.

\begin{lemma}\label{lem:condition}
Consider an instance of the quantum quadratic ODE problem as defined in \prb{node}. Apply the forward Euler method \eq{forward} with time step \eq{h} to the Carleman linearization \eq{LODE}. Then the condition number of the matrix $L$ defined in \eq{linear_system} satisfies
\begin{equation}
\kappa \le 3(m+p+1).
\end{equation}
\end{lemma}

\begin{proof}
We begin by upper bounding $\|L\|$. We write
\begin{equation}
L = L_1 + L_2 + L_3,
\end{equation}
where
\begin{align}
L_1 &= \sum_{k=0}^{m+p}|k\rangle\langle k|\otimes I, \\
L_2 &= -\sum_{k=1}^{m}|k\rangle\langle k-1|\otimes [I+A((k-1)h)h], \\
L_3 &= -\sum_{k=m+1}^{m+p}|k\rangle\langle k-1|\otimes I.
\end{align}
Clearly $\|L_1\| = \|L_3\| = 1$. Furthermore, $\|L_2\|\le \max_{t\in[0,T]}\|I+A(t)h\| \le 1$ by \eq{norm_equivalent}, which follows from the choice of time step \eq{h}. Therefore,
\begin{equation}
\|L\| \le \|L_1\| + \|L_2\| + \|L_3\| \le 3.
\label{eq:upper}
\end{equation}

Next we upper bound 
\begin{equation}
\|L^{-1}\|=\sup_{\||B\rangle\|\le1}\|L^{-1}|B\rangle\|.
\end{equation}
We express $|B\rangle$ as
\begin{equation}
|B\rangle= \sum_{k=0}^{m+p}\beta_k|k\rangle = \sum_{k=0}^{m+p}|b^k\rangle,
\label{eq:Bvec}
\end{equation}
where $|b^k\rangle \coloneqq \beta_k|k\rangle$ satisfies
\begin{equation}
\sum_{k=0}^{m+p}\||b^k\rangle\|^2 = \||B\rangle\|^2 \le 1.
\end{equation}
Given any $|b^k\rangle$ for $k\in\rangez{m+p+1}$, we define
\begin{equation}
|Y^k\rangle \coloneqq L^{-1}|b^k\rangle = \sum_{l=0}^{m+p}\gamma^k_l|l\rangle = \sum_{l=0}^{m+p}|Y^k_l\rangle,
\label{eq:Yvec}
\end{equation}
where $|Y^k_l\rangle \coloneqq \gamma^k_l|l\rangle$. We first upper bound $\||Y^k\rangle\| = \|L^{-1}|b^k\rangle\|$, and then use this to upper bound $\|L^{-1}|B\rangle\|$.

We consider two cases. First, for fixed $k\in\rangez{m+1}$, we directly calculate $|Y^k_l\rangle$ for each $l\in\rangez{m+p+1}$ by the forward Euler method \eq{forward}, giving
\begin{equation}
|Y^k_l\rangle =
\begin{cases}
0, &\text{if } l\in\rangez{k}; \\
|b^k\rangle, &\text{if } l=k; \\
\Pi_{j=k}^{l-1}[I+A(jh)h]|b^k\rangle, &\text{if } l\in\rangez{m+1}\setminus\rangez{k+1}; \\
\Pi_{j=k}^{m-1}[I+A(jh)h]|b^k\rangle, &\text{if } l\in\rangez{m+p+1}\setminus\rangez{m+1}.
\label{eq:expand_1}
\end{cases}
\end{equation}
Since $\max_{t\in[0,T]}\|I+A(t)h\|| \le 1$ by \eq{norm_equivalent}, \eq{Yvec} gives
\begin{equation}
\begin{aligned}
\||Y^k\rangle\|^2
&= \sum_{l=0}^{m+p}\||Y^k_l\rangle\|^2
 \le \sum_{l=k}^{m}\||b^k\rangle\|^2 + \sum_{l=m+1}^{m+p}\||b^k\rangle\|^2 \\
&\le (m+p+1-k)\||b^k\rangle\|^2 \le (m+p+1)\||b^k\rangle\|^2.
\label{eq:expand_norm_1}
\end{aligned}
\end{equation}

Second, for fixed $k\in\rangez{m+p+1}\setminus\rangez{m+1}$, similarly to \eq{expand_1}, we directly calculate $|Y^k_l\rangle$ using \eq{forward}, giving
\begin{equation}
|Y^k_l\rangle =
\begin{cases}
0, &\text{if } l\in\rangez{k}; \\
|b^k\rangle, &\text{if } l\in\rangez{m+p+1}\setminus\rangez{k}.
\label{eq:expand_2}
\end{cases}
\end{equation}
Similarly to \eq{expand_norm_1}, we have (again using \eq{Yvec})
\begin{equation}
\||Y^k\rangle\|^2 = \sum_{l=0}^{m+p}\||Y^k_l\rangle\|^2 = \sum_{l=k}^{m+p}\||b^k\rangle\|^2 = (m+p+1-k)\||b^k\rangle\|^2 \le (m+p+1)\||b^k\rangle\|^2.
\label{eq:expand_norm_2}
\end{equation}

Combining \eq{expand_norm_1} and \eq{expand_norm_2}, for any $k\in\rangez{m+p+1}$, we have
\begin{equation}
\||Y^k\rangle\|^2 = \|L^{-1}|b^k\rangle\|^2 \le (m+p+1)\||b^k\rangle\|^2.
\label{eq:expand_norm_Y}
\end{equation}
By the definition of $|Y^k\rangle$ in \eq{Yvec}, \eq{expand_norm_Y} gives
\begin{equation}
\| L^{-1}\|^2 = \sup_{\||B\rangle\|\le1}\| L^{-1}|B\rangle\|^2 \le (m+p+1) \sup_{\||b^k\rangle\|\le1}\|L^{-1}|b^k\rangle\|^2 \le (m+p+1)^2,
\end{equation}
and therefore
\begin{equation}
\| L^{-1}\|  \le (m+p+1).
\label{eq:lower}
\end{equation}
Finally, combining \eq{upper} with \eq{lower}, we conclude
\begin{equation}
\kappa = \| L\|\| L^{-1}\| \le 3(m+p+1)
\end{equation}
as claimed.
\end{proof}

\subsection{State preparation}

We now describe a procedure for preparing the right-hand side $|B\rangle$ of the linear system \eq{linear_system}, whose entries are composed of $|y_{\mathrm{in}}\rangle$ and $|b((k-1)h)\rangle$ for $k\in\range{m}$.

The initial vector $y_{\mathrm{in}}$ is a direct sum over spaces of different dimensions, which is cumbersome to prepare. Instead, we prepare an equivalent state that has a convenient tensor product structure. Specifically, we embed $y_{\mathrm{in}}$ into a slightly larger space and prepare the normalized version of
\begin{equation}
    z_{\mathrm{in}} = [u_{\mathrm{in}}\otimes v_0^{N-1}; u_{\mathrm{in}}^{\otimes 2}\otimes v_0^{N-2}; \ldots; u_{\mathrm{in}}^{\otimes N}],
\end{equation}
where $v_0$ is some standard vector (for simplicity, we take $v_0 = \ket{0}$). If $u_{\mathrm{in}}$ lives in a vector space of dimension $n$, then $z_{\mathrm{in}}$ lives in a space of dimension $Nn^N$ while $y_{\mathrm{in}}$ lives in a slightly smaller space of dimension $\dimn = n+n^2+\dots +n^N = (n^{N+1}-n)/(n-1)$. Using standard techniques, all the operations we would otherwise apply to $y_{\mathrm{in}}$ can be applied instead to $z_{\mathrm{in}}$, with the same effect.

\begin{lemma}\label{lem:preparation}
Assume we are given the value $\|u_{\mathrm{in}}\|$, and let $O_x$ be an oracle that maps $|00\ldots0\rangle \in\C^n$ to a normalized quantum state $|u_{\mathrm{in}}\rangle$ proportional to $u_{\mathrm{in}}$. 
Assume we are also given the values $\|F_0(t)\|$ for each $t\in[0,T]$, and let $O_{F_0}$ be an oracle that provides the locations and values of the nonzero entries of $F_0(t)$ for any specified $t$. Then the quantum state $|B\rangle$ defined in \eq{vectorB} (with $y_{\mathrm{in}}$ replaced by $z_{\mathrm{in}}$) can be prepared using $O(N)$ queries to $O_x$ and $O(m)$ queries to $O_{F_0}$,  with gate complexity larger by a $\poly(\log N, \log n)$ factor.
\end{lemma}

\begin{proof}
We first show how to prepare the state
\begin{equation}
    \ket{z_{\mathrm{in}}}=\frac{1}{\sqrt{V}}\sum_{j=1}^N \norm{u_{\mathrm{in}}}^j |j\rangle|u_{\mathrm{in}}\rangle^{\otimes j} |0\rangle^{\otimes N-j},
\end{equation}
where
\begin{equation}
    V \coloneqq \sum_{j=1}^{N}\|u_{\mathrm{in}}\|^{2j}.
\end{equation}
This state can be prepared using $N$ queries to the initial state oracle $O_x$ applied in superposition to the intermediate state
\begin{equation}
    \ket{\psi_{\mathrm{int}}}\coloneqq\frac{1}{\sqrt{V}}\sum_{j=1}^N \norm{u_{\mathrm{in}}}^j |j\rangle\otimes \ket{0}^{\otimes N}.
\end{equation}

To efficiently prepare $\ket{\psi_{\mathrm{int}}}$, notice that 
\begin{equation}
    \ket{\psi_{\mathrm{int}}}=\frac{\norm{u_{\mathrm{in}}}}{\sqrt{V}}\sum_{j_0,j_1,\dots, j_{k-1}=0}^1 \prod_{\ell=0}^{k-1} \norm{u_{\mathrm{in}}}^{j_\ell2^\ell}\ket{j_0j_1\dots j_{k-1}} \otimes \ket{0}^{\otimes N},
\end{equation}
where $k \coloneqq \log_2 N$ (assuming for simplicity that $N$ is a power of $2$) and $j_{k-1} \dots j_1 j_0$ is the $k$-bit binary expansion of $j-1$. Observe that
\begin{equation}\label{eq:psi_int}
    \ket{\psi_{\mathrm{int}}}=\bigotimes_{\ell=0}^{k-1} \bigg(\frac{1}{\sqrt{V_\ell}}\sum_{j_\ell=0}^1 \norm{u_{\mathrm{in}}}^{j_\ell2^\ell}\ket{j_\ell}\bigg)\otimes \ket{0}^{\otimes N}
\end{equation}
where 
\begin{equation}
    V_\ell \coloneqq 1+\norm{u_{\mathrm{in}}}^{2^{\ell+1}}.
\end{equation}
(Notice that $\prod_{\ell=0}^{k-1} V_\ell=V/\norm{u_{\mathrm{in}}}^2$.) Each tensor factor in \eq{psi_int} is a qubit that can be produced in constant time.
Overall, we prepare these $k=\log_2 N$ qubit states and then apply $O_x$ $N$ times.

We now discuss how to prepare the state
\begin{equation}
|B\rangle = \frac{1}{\sqrt{B_m}}\Bigl(\|z_{\mathrm{in}}\||0\rangle \otimes |z_{\mathrm{in}}\rangle+\sum_{k=1}^{m}\|b((k-1)h)\||k\rangle \otimes |b((k-1)h)\rangle\Bigr),
\end{equation}
in which we replace $y_{\mathrm{in}}$ by $z_{\mathrm{in}}$ in \eq{vectorB}, and define
\begin{equation}
    B_m \coloneqq \norm{z_{\mathrm{in}}}^2 + \sum_{k=1}^m\|b((k-1)h)\|^2.
\end{equation}
This state can be prepared using the above procedure for $|0\rangle \mapsto |z_{\mathrm{in}}\rangle$ and $m$ queries to $O_{F_0}$ with $t=(k-1)h$ that implement $|0\rangle \mapsto |b((k-1)h)\rangle$ for $k \in \{1,\ldots,m\}$, applied in superposition to the intermediate state
\begin{equation}
|\phi_{\mathrm{int}}\rangle = \frac{1}{\sqrt{B_m}}\Bigl(\|z_{\mathrm{in}}\||0\rangle \otimes |0\rangle+\sum_{k=1}^{m}\|b((k-1)h)\||k\rangle \otimes |0\rangle\Bigr).
\end{equation}
Here the queries are applied conditionally upon the value in the first register: we prepare $|z_{\mathrm{in}}\rangle$ if the first register is $|0\rangle$ and $|b((k-1)h)\rangle$ if the first register is $|k\rangle$ for $k \in \{1,\ldots,m\}$.
We can prepare $|\phi_{\mathrm{int}}\rangle$ (i.e., perform a unitary transform mapping $|0\rangle|0\rangle \mapsto |\phi_{\mathrm{int}}\rangle$) in time complexity $O(m)$ \cite{SBM06} using the known values of $\norm{u_{\mathrm in}}$ and $\norm{b((k-1)h)}$. 

Overall, we use $O(N)$ queries to $O_x$ and $O(m)$ queries to $O_{F_0}$ to prepare $|B\rangle$. The gate complexity is larger by a $\poly(\log N, \log n)$ factor.
\end{proof}

\subsection{Measurement success probability}

After applying the QLSA to \eq{linear_system}, we perform a measurement to extract a final state of the desired form. We now consider the probability of this measurement succeeding.

\begin{lemma}\label{lem:measure}
Consider an instance of the quantum quadratic ODE problem defined in \prb{node}, with the QLSA applied to the linear system \eq{linear_system} using the forward Euler method \eq{forward} with time step \eq{h}. Suppose the error from Carleman linearization satisfies $\|\eta(t)\|\le \frac{g}{4}$ as in \eq{g_condition_1}, and the global error from the forward Euler method as defined in \lem{global} is bounded by
\begin{equation}
\|\hat y(T)-y^m\| \le \frac{g}{4},
\label{eq:g_condition_2}
\end{equation}
where $g$ is defined in \eq{gq}. Then the probability of measuring a state $|y^k_1\rangle$ for $k=\rangez{m+p+1}\setminus\rangez{m+1}$ satisfies
\begin{equation}
P_{\mathrm{measure}} \ge \frac{p+1}{9(m+p+1)Nq^2},
\end{equation}
where $q$ is also defined in \eq{gq}.
\end{lemma}

\begin{proof}
The idealized quantum state produced by the QLSA applied to \eq{linear_system} has the form
\begin{equation}
|Y\rangle = \sum_{k=0}^{m+p}|y^k\rangle|k\rangle = \sum_{k=0}^{m+p}\sum_{j=1}^N|y_j^k\rangle|j\rangle|k\rangle
\end{equation}
where the states $|y^k\rangle$ and $|y_j^k\rangle$ for $k \in \rangez{m+p+1}$ and $j \in \range{N}$ are subnormalized to ensure $\||Y\rangle\|=1$.

We decompose the state $|Y\rangle$ as
\begin{equation}
|Y\rangle = |Y_{\mathrm{bad}}\rangle+|Y_{\mathrm{good}}\rangle,
\end{equation}
where
\begin{equation}
\begin{aligned}
|Y_{\mathrm{bad}}\rangle &\coloneqq \sum_{k=0}^{m-1}\sum_{j=1}^N|y_j^k\rangle|j\rangle|k\rangle+\sum_{k=m}^{m+p}\sum_{j=2}^N|y_j^k\rangle|j\rangle|k\rangle, \\
|Y_{\mathrm{good}}\rangle &\coloneqq \sum_{k=m}^{m+p}|y_1^k\rangle|1\rangle|k\rangle.
\end{aligned}
\end{equation}
Note that $|y_1^k\rangle = |y_1^m\rangle$ for all $k \in \{m,m+1,\ldots,m+p\}$.
We lower bound
\begin{equation}
P_{\mathrm{measure}} \coloneqq \frac{\||Y_{\mathrm{good}}\rangle\|^2}{\||Y\rangle\|^2} = \frac{(p+1)\||y^m_1\rangle\|^2}{\||Y\rangle\|^2}
\label{eq:measure_good}
\end{equation}
by lower bounding the terms of the product
\begin{equation}
\frac{\||y^m_1\rangle\|^2}{\||Y\rangle\|^2} = \frac{\||y^m_1\rangle\|^2}{\||y^0_1\rangle\|^2} \cdot \frac{\||y^0_1\rangle\|^2}{\||Y\rangle\|^2}.
\label{eq:measure_relation}
\end{equation}

First, according to \eq{g_condition_1} and \eq{g_condition_2}, the exact solution $u(T)$ and the approximate solution $y^m_1$ defined in \eq{linear_system} satisfy
\begin{equation}
\|u(T)-y^m_1\| \le \|u(T)-\hat y_1(T)\|+\|\hat y_1(T)-y^m_1\| \le \|\eta(t)\|+\|\hat y(T)-y^m\| \le \frac{g}{2}.
\label{eq:epsilon_inequality}
\end{equation}
Since $y^0_1=(y_{\mathrm{in}})_1=u_{\mathrm{in}}$, using \eq{epsilon_inequality}, we have
\begin{equation}
\frac{\||y^m_1\rangle\|}{\||y^0_1\rangle\|} 
= \frac{\|y^m_1\|}{\|u_{\mathrm{in}}\|} 
\ge \frac{\|u(T)\| -\|u(T)-y^m_1\|}{\|u_{\mathrm{in}}\|} 
= \frac{g-\|u(T)-y^m_1\|}{\|u_{\mathrm{in}}\|} 
\ge \frac{g}{2\|u_{\mathrm{in}}\|} 
= \frac{1}{2q}.
\label{eq:measure_1}
\end{equation}

Second, we upper bound $\|y^k\|^2$ by
\begin{equation}
\begin{aligned}
\|y^k\|^2 = \|\hat y(kh) - [y(kh)-y^k]\|^2
\le 2(\|\hat y(kh)\|^2 + \|y(kh)-y^k\|^2).
\end{aligned}
\end{equation}
Using $\|\hat y(t)\|^2 < 4N\|u_{\mathrm{in}}\|^2$ by \eq{parallel_inequality}, and $\|y(kh)-y^k\| \le \|\hat y(T)-y^m\| \le g/4 < \|u_{\mathrm{in}}\|/4$ by \eq{g_condition_2}, and $\Rq < 1$, we have
\begin{equation}
\begin{aligned}
\|y^k\|^2 \le 2\biggl(4N\|u_{\mathrm{in}}\|^2 + \frac{\|u_{\mathrm{in}}\|^2}{16}\biggr)
< 9N\|u_{\mathrm{in}}\|^2.
\end{aligned}
\end{equation}
Therefore
\begin{equation}
  \frac{\||y^0_1\rangle\|^2}{\||Y\rangle\|^2} = \frac{\||y^0_1\rangle\|^2}{\sum_{k=0}^{m+p}\||y^k\rangle\|^2} \ge \frac{\|u_{\mathrm{in}}\|^2}{9N(m+p+1)\|u_{\mathrm{in}}\|^2} = \frac{1}{9N(m+p+1)}.
  \label{eq:measure_2}
\end{equation}

Finally, using \eq{measure_1} and \eq{measure_2} in \eq{measure_relation} and \eq{measure_good}, we have
\begin{equation}
P_{\mathrm{measure}} \ge \frac{p+1}{9(m+p+1)Nq^2}
\end{equation}
as claimed.
\end{proof}

Choosing $m=p$, we have $P_{\mathrm{measure}} = \Omega(1/Nq^2)$. Using amplitude amplification, $O(\sqrt{N}q)$ iterations suffice to succeed with constant probability.

\subsection{Proof of \texorpdfstring{\thm{main}}{Theorem \ref{thm:main}}}

\begin{proof}
We first present the quantum Carleman linearization (QCL) algorithm and then analyze its complexity.

\paragraph{The QCL algorithm.} We start by rescaling the system to satisfy \eq{A2} and \eq{A3}. Given a quadratic ODE \eq{NODE} satisfying $\Rq<1$ (where $\Rq$ is defined in \eq{A1}), we define a scaling factor $\gamma>0$, and rescale $u$ to
\begin{equation}
\overline{u} \coloneqq \gamma u.
\end{equation}
Replacing $u$ by $\overline{u}$ in \eq{NODE}, we have
\begin{equation}
\begin{aligned}
\frac{\d{\overline{u}}}{\d{t}} &= \frac{1}{\gamma}F_2\overline{u}^{\otimes 2}+F_1\overline{u}+\gamma F_0(t), \\
\overline{u}(0) &= \overline{u}_{\mathrm{in}} \coloneqq \gamma u_{\mathrm{in}}.
\label{eq:NNODE}
\end{aligned}
\end{equation}
We let $\overline{F}_2 \coloneqq \frac{1}{\gamma}F_2$, $\overline{F}_1 \coloneqq F_1$, and $\overline{F}_0(t) \coloneqq \gamma F_0(t)$ so that
\begin{equation}
\begin{aligned}
\frac{\d{\overline{u}}}{\d{t}} &= \overline{F}_2\overline{u}^{\otimes 2}+\overline{F}_1\overline{u}+\overline{F}_0(t), \\
\overline{u}(0) &= \overline{u}_{\mathrm{in}}.
\end{aligned}
\end{equation}
Note that $\Rq$ is invariant under this rescaling, so $\Rq<1$ still holds for the rescaled equation.

Concretely, we take\footnote{In fact, one can show that any $\gamma \in \bigl(\frac{1}{r_+},\frac{1}{\|u_{\mathrm{in}}\|}\bigr)$ suffices to satisfy \eq{A2} and \eq{A3}.\label{foot:varygamma}}
\begin{equation}
\gamma = \frac{1}{\sqrt{\|u_{\mathrm{in}}\|r_+}}.
\label{eq:rescaling1}
\end{equation}
After rescaling, the new quadratic ODE satisfies $\|\overline{u}_{\mathrm{in}}\|\overline{r}_+=\gamma^2\|u_{\mathrm{in}}\|r_+=1$. Since $\|u_{\mathrm{in}}\|<r_+$ by \lem{decrease}, we have $\overline{r}_-<\|\overline{u}_{\mathrm{in}}\|<1<\overline{r}_+$, so \eq{A3} holds. Furthermore, $1$ is located between the two roots $\overline{r}_-$ and $\overline{r}_+$, which implies $\|\overline{F}_2\|\cdot1^2 + |{\Re{(\overline{\lambda}_1)}}|\cdot1 + \|\overline{F}_0\| < 0$ as shown in \lem{decrease}, so \eq{A2} holds for the rescaled problem.

Having performed this rescaling, we henceforth assume that \eq{A2} and \eq{A3} are satisfied. We then introduce the choice of parameters as follows. Given $g$ and an error bound $\epsilon\le1$, we define
\begin{equation}
\delta \coloneqq \frac{g\epsilon}{1+\epsilon} \le \frac{g}{2}.
\label{eq:delta_1}
\end{equation}
Given $\|u_{\mathrm{in}}\|$, $\|F_2\|$, and $\Re(\lambda_1)<0$, we choose
\begin{equation}
N = \biggl\lceil \frac{\log(2T\|F_2\|/\delta)}{\log(1/\|u_{\mathrm{in}}\|)} \biggr\rceil = \biggl\lceil \frac{\log(2T\|F_2\|/\delta)}{\log(r_+)} \biggr\rceil.
\label{eq:N}
\end{equation}
Since $\|u_{\mathrm{in}}\|/\delta>1$ by \eq{delta_1} and $g < \|u_{\mathrm{in}}\|$, 
\lem{Carleman} gives
\begin{equation}
\|u(T)-\hat y_1(T)\| \le \|\eta(T)\| \le TN\|F_2\|\|u_{\mathrm{in}}\|^{N+1} = TN\|F_2\|(\frac{1}{r_+})^{N+1} \le \frac{\delta}{2}.
\label{eq:error_1}
\end{equation}
Thus, \eq{g_condition_1} holds since $\delta \le g/2$.

Now we discuss the choice of $h$. On the one hand, $h$ must satisfy \eq{h} to satisfy the conditions of \lem{global} and \lem{condition}.
On the other hand, \lem{global} gives the upper bound 
\begin{equation}
\|\hat y_1(T)-y^m_1\| \le 3N^{2.5}Th[(\|F_2\|+\|F_1\|+\|F_0\|)^2+\|F_0'\|] \le \frac{g\epsilon}{4} \le \frac{g\epsilon}{2(1+\epsilon)} = \frac{\delta}{2}.
\label{eq:error_2}
\end{equation}
This also ensures that \eq{g_condition_2} holds since $\delta \le g/2$. Thus, we choose
\begin{equation}
\begin{aligned}
h \le \min \biggl\{& \frac{g\epsilon}{12N^{2.5}T[(\|F_2\|+\|F_1\|+\|F_0\|)^2+\|F_0'\|]},\frac{1}{N\|F_1\|}, \\ & \frac{2(|\Re(\lambda_1)|-\|F_2\|-\|F_0\|)}{N(|\Re(\lambda_1)|^2 - (\|F_2\|+\|F_0\|)^2 + \|F_1\|^2)} \biggr\}
\end{aligned}
\end{equation}
to satisfy \eq{h} and \eq{error_2}.

Combining \eq{error_1} with \eq{error_2}, we have
\begin{equation}
\|u(T)-y^m_1\| \le \|u(T)-\hat y_1(T)\|+\|\hat y_1(T)-y^m_1\| \le \delta.
\label{eq:error_3}
\end{equation}
Thus, \eq{epsilon_inequality} holds since $\delta \le g/2$. Using 
\begin{equation}
\biggl\|\frac{u(T)}{\|u(T)\|}-\frac{y^m_1}{\|y^m_1\|}\biggr\| \le \frac{\|u(T)-y^m_1\|}{\min\{\|u(T)\|,\|y^m_1\|\}} \le \frac{\|u(T)-y^m_1\|}{g-\|u(T)-y^m_1\|}
\label{eq:normalized_inequality}
\end{equation}
and \eq{error_3}, we obtain
\begin{equation}
\biggl\|\frac{u(T)}{\|u(T)\|}-\frac{y^m_1}{\|y^m_1\|}\biggr\| \le \frac{\delta}{g-\delta} = \epsilon,
\label{eq:error_4}
\end{equation}
i.e., the normalized output state is $\epsilon$-close to $\frac{u(T)}{\|u(T)\|}$.

We follow the procedure in \lem{preparation} to prepare the initial state $|\hat y_{\mathrm{in}}\rangle$.
We apply the QLSA \cite{CKS15} to the linear system \eq{linear_system} with $m=p=\lceil T/h \rceil$, giving a solution $|Y\rangle$. We then perform a measurement to obtain a normalized state of $|y^k_j\rangle$ for some 
$k\in\rangez{m+p+1}$ and $j\in\range{N}$.
By \lem{measure}, the probability of obtaining a state $|y^k_1\rangle$ for some $k\in\rangez{m+p+1}\setminus\rangez{m+1}$, giving the normalized vector ${y^m_1}/{\|y^m_1\|}$, is
\begin{equation}
P_{\mathrm{measure}} \ge \frac{p+1}{9(m+p+1)Nq^2} \ge \frac{1}{18Nq^2}.
\end{equation}
By amplitude amplification, we can achieve success probability $\Omega(1)$ with $O(\sqrt{N}q)$ repetitions of the above procedure.

\paragraph{Analysis of the complexity.} By \lem{preparation}, the right-hand side $|B\rangle$ in \eq{linear_system} can be prepared with $O(N)$ queries to $O_x$ and $O(m)$ queries to $O_{F_0}$, with gate complexity larger by a $\poly(\log N, \log n)$ factor. The matrix $L$ in \eq{linear_system} is an $(m+p+1)\dimn\times(m+p+1)\dimn$ matrix with $O(Ns)$ nonzero entries in any row or column. By \lem{condition} and our choice of parameters, the condition number of $L$ is at most
\begin{equation}
\begin{aligned}
&3(m+p+1) \\
&= O\biggl(\frac{N^{2.5}T^2[(\|F_2\|+\|F_1\|+\|F_0\|)^2+\|F_0'\|]}{\delta}+NT\|F_1\| \\
&\qquad\quad +\frac{NT[|\Re(\lambda_1)|^2 - (\|F_2\|+\|F_0\|)^2 + \|F_1\|^2]}{2(|\Re(\lambda_1)|-\|F_2\|-\|F_0\|)} \biggr) \\
&= O\biggl( \frac{N^{2.5}T^2[(\|F_2\|+\|F_1\|+\|F_0\|)^2+\|F_0'\|]}{g\epsilon} + \frac{1}{(1-\|u_{\mathrm{in}}\|)^2}\cdot\frac{NT\|F_1\|^2}{\|F_2\|+\|F_0\|} \biggr) \\
&= O\biggl(\frac{N^{2.5}T^2[(\|F_2\|+\|F_1\|+\|F_0\|)^2+\|F_0'\|]}{(1-\|u_{\mathrm{in}}\|)^2(\|F_2\|+\|F_0\|)g\epsilon}\biggr).
\label{eq:condition_estimate}
\end{aligned}
\end{equation}
Here we use $\|F_2\|+\|F_0\|<|\Re(\lambda_1)|\le\|F_1\|$ and
\begin{equation}
\begin{aligned}
2&(|\Re(\lambda_1)|-\|F_2\|-\|F_0\|)
>(\|u_{\mathrm{in}}\|+\frac{1}{\|u_{\mathrm{in}}\|}-2)(\|F_2\|+\|F_0\|) \\
&=\frac{1}{\|u_{\mathrm{in}}\|}(1-\|u_{\mathrm{in}}\|)^2(\|F_2\|+\|F_0\|)
> (1-\|u_{\mathrm{in}}\|)^2(\|F_2\|+\|F_0\|).
\end{aligned}
\end{equation}
The first inequality above is from the sum of $|\Re(\lambda_1)|>\|F_2\|\|u_{\mathrm{in}}\|+\|F_0\|/\|u_{\mathrm{in}}\|$ and $|\Re(\lambda_1)|=\|F_2\|r_++\|F_0\|/r_+$, where $r_+=1/\|u_{\mathrm{in}}\|$. Consequently, by Theorem 5 of \cite{CKS15}, the QLSA produces the state $|Y\rangle$ with
\begin{equation}
\begin{aligned}
& \frac{N^{3.5}sT^2[(\|F_2\|+\|F_1\|+\|F_0\|)^2+\|F_0'\|]}{(1-\|u_{\mathrm{in}}\|)^2(\|F_2\|+\|F_0\|)g\epsilon}\cdot
\poly\biggl(\log\frac{NsT\|F_2\|\|F_1\|\|F_0\|\|F'_0\|}{(1-\|u_{\mathrm{in}}\|)g\epsilon}\biggr) \\
&\quad = \frac{sT^2[(\|F_2\|+\|F_1\|+\|F_0\|)^2+\|F_0'\|]}{(1-\|u_{\mathrm{in}}\|)^2(\|F_2\|+\|F_0\|)g\epsilon}\cdot
\poly\biggl(\log\frac{sT\|F_2\|\|F_1\|\|F_0\|\|F'_0\|}{(1-\|u_{\mathrm{in}}\|)g\epsilon}/\log (1/\|u_{\mathrm{in}}\|)\biggr)
\end{aligned}
\end{equation}
queries to the oracles $O_{F_2}$, $O_{F_1}$, $O_{F_0}$, and $O_{x}$. Using $O(\sqrt{N}q)$ steps of amplitude amplification to achieve success probability $\Omega(1)$, the overall query complexity of our algorithm is
\begin{equation}
\begin{aligned}
& \frac{N^4sT^2q[(\|F_2\|+\|F_1\|+\|F_0\|)^2+\|F_0'\|]}{(1-\|u_{\mathrm{in}}\|)^2(\|F_2\|+\|F_0\|)g\epsilon}\cdot
\poly\biggl(\log\frac{NsT\|F_2\|\|F_1\|\|F_0\|\|F'_0\|}{(1-\|u_{\mathrm{in}}\|)g\epsilon}\biggr) \\
&\quad = \frac{sT^2q[(\|F_2\|+\|F_1\|+\|F_0\|)^2+\|F_0'\|]}{(1-\|u_{\mathrm{in}}\|)^2(\|F_2\|+\|F_0\|)g\epsilon}\cdot
\poly\biggl(\log\frac{sT\|F_2\|\|F_1\|\|F_0\|\|F'_0\|}{(1-\|u_{\mathrm{in}}\|)g\epsilon}/\log (1/\|u_{\mathrm{in}}\|)\biggr)
\label{eq:query_complexity}
\end{aligned}
\end{equation}
and the gate complexity exceeds this by a factor of
$\poly\bigl(\log(nsT\|F_2\|\|F_1\|\|F_0\|\|F'_0\|/(1-\Rq) g\epsilon)/\allowbreak \log (1/\|u_{\mathrm{in}}\|)\bigr)$.

If the eigenvalues $\lambda_j$ of $F_1$ are all real, by \eq{h_real}, the condition number of $L$ is at most
\begin{equation}
\begin{aligned}
3(m+p+1) &= O\biggl(\frac{N^{2.5}T^2[(\|F_2\|+\|F_1\|+\|F_0\|)^2+\|F_0'\|]}{\delta}+NT\|F_1\|\biggr) \\
&= O\biggl(\frac{N^{2.5}T^2[(\|F_2\|+\|F_1\|+\|F_0\|)^2+\|F_0'\|]}{g\epsilon}\biggr).
\end{aligned}
\end{equation}
Similarly, the QLSA produces the state $|Y\rangle$ with
\begin{equation}
\begin{aligned}
\frac{sT^2[(\|F_2\|+\|F_1\|+\|F_0\|)^2+\|F_0'\|]}{g\epsilon}\cdot
\poly\biggl(\log\frac{sT\|F_2\|\|F_1\|\|F_0\|\|F'_0\|}{g\epsilon}/\log (1/\|u_{\mathrm{in}}\|) \biggr)
\end{aligned}
\end{equation}
queries to the oracles $O_{F_2}$, $O_{F_1}$, $O_{F_0}$, and $O_{x}$. Using amplitude amplification to achieve success probability $\Omega(1)$, the overall query complexity of the algorithm is
\begin{equation}
\begin{aligned}
\frac{sT^2q[(\|F_2\|+\|F_1\|+\|F_0\|)^2+\|F_0'\|]}{g\epsilon}\cdot
\poly\biggl(\log\frac{sT\|F_2\|\|F_1\|\|F_0\|\|F'_0\|}{g\epsilon}/\log (1/\|u_{\mathrm{in}}\|) \biggr)
\label{eq:query_complexity_real}
\end{aligned}
\end{equation}
and the gate complexity is larger by a factor of
$\poly\bigl(\log(nsT\|F_2\|\|F_1\|\|F_0\|\|F'_0\|/g\epsilon)/\log (1/\|u_{\mathrm{in}}\|)\bigr)$
as claimed.
\end{proof}

\section{Lower bound}
\label{sec:lower}

In this section, we establish a limitation on the ability of quantum computers to solve the quadratic ODE problem when the nonlinearity is sufficiently strong. We quantify the strength of the nonlinearity in terms of the quantity $\Rq$ defined in \eq{A1}. Whereas there is an efficient quantum algorithm for $\Rq < 1$ (as shown in \thm{main}), we show here that the problem is intractable for $\Rq \ge \sqrt2$.

\begin{theorem}\label{thm:lower}
Assume $\Rq \ge \sqrt2$. Then there is an instance of the quantum quadratic ODE problem defined in \prb{node} such that any quantum algorithm for producing a quantum state approximating $u(T)/\| u(T)\|$ with bounded error must have worst-case time complexity exponential in $T$.
\end{theorem}

We establish this result by showing how the nonlinear dynamics can be used to distinguish nonorthogonal quantum states, a task that requires many copies of the given state. Note that since our algorithm only approximates the quantum state corresponding to the solution, we must lower bound the query complexity of \emph{approximating} the solution of a quadratic ODE.

\subsection{Hardness of state discrimination}


Previous work on the computational power of nonlinear quantum mechanics shows that the ability to distinguish non-orthogonal states can be applied to solve unstructured search (and other hard computational problems) \cite{AL98,Aar05,CY16}. Here we show a similar limitation using an information-theoretic argument.

\begin{lemma}\label{lem:search}
Let $|\psi\rangle,|\phi\rangle$ be states of a qubit with $|\langle\psi|\phi\rangle|=1-\epsilon$. Suppose we are either given a black box that prepares $|\psi\rangle$ or a black box that prepares $|\phi\rangle$. Then any bounded-error protocol for determining whether the state is $|\psi\rangle$ or $|\phi\rangle$ must use $\Omega(1/\epsilon)$ queries.
\end{lemma}

\begin{proof}
Using the black box $k$ times, we prepare states with overlap $(1-\epsilon)^k$. By the well-known relationship between fidelity and trace distance, these states have trace distance at most $\sqrt{1-(1-\epsilon)^{2k}} \le \sqrt{2k\epsilon}$. Therefore, by the Helstrom bound (which states that the advantage over random guessing for the best measurement to distinguish two quantum states is given by their trace distance \cite{Hel69}), we need $k = \Omega(1/\epsilon)$ to distinguish the states with bounded error.
\end{proof}

\subsection{State discrimination with nonlinear dynamics}

\lem{search} can be used to establish limitations on the ability of quantum computers to simulate nonlinear dynamics, since nonlinear dynamics can be used to distinguish nonorthogonal states. Whereas previous work considers models of nonlinear quantum dynamics (such as the Weinberg model \cite{AL98,Aar05} and the Gross-Pitaevskii equation \cite{CY16}), here we aim to show the difficulty of efficiently simulating more general nonlinear ODEs---in particular, quadratic ODEs with dissipation---using quantum algorithms.

\begin{lemma}\label{lem:discrimination}
There exists an instance of the quantum quadratic ODE problem as defined in \prb{node} with $\Rq \ge \sqrt2$, and two states of a qubit with overlap $1-\epsilon$ (for $0 < \epsilon < 1-3/\sqrt{10}$) as possible initial conditions, such that the two final states after evolution time $T=O(\log(1/\epsilon))$ have an overlap no larger than $3/\sqrt{10}$.
\end{lemma}

\begin{proof}
Consider a $2$-dimensional system of the form
\begin{equation}
\begin{aligned}
\frac{\d{u_1}}{\d{t}} &= -u_1 + r u_1^2, \\
\frac{\d{u_2}}{\d{t}} &= -u_2 + r u_2^2,
\end{aligned}
\label{eq:NODE2}
\end{equation}
for some $r>0$, with an initial condition $u(0)=[u_1(0);u_2(0)]=u_{\mathrm{in}}$ satisfying $\|u_{\mathrm{in}}\|=1$. According to the definition of $\Rq$ in \eq{A1}, we have $\Rq=r$, so henceforth we write this parameter as $\Rq$. The analytic solution of \eq{NODE2} is
\begin{equation}
\begin{aligned}
u_1(t) &= \frac{1}{\Rq-e^t(\Rq-1/u_1(0))}, \\
u_2(t) &= \frac{1}{\Rq-e^t(\Rq-1/u_2(0))}.
\end{aligned}
\label{eq:u1u2}
\end{equation}

When $u_2(0)>1/\Rq$, $u_2(t)$ is finite within the domain
\begin{equation}
0 \le t < t^{\ast} \coloneqq \log\biggl(\frac{\Rq}{\Rq-1/u_2(0)}\biggr);
\label{eq:domain}
\end{equation}
when $u_2(0)=1/\Rq$, we have $u_2(t)=1/\Rq$ for all $t$; and when $u_2(0)<1/\Rq$, $u_2(t)$ goes to $0$ as $t \to \infty$. The behavior of $u_1(t)$ depends similarly on $u_1(0)$.

Without loss of generality, we assume $u_1(0) \le u_2(0)$. For $u_2(0)\ge u_1(0)>1/\Rq$, both $u_1(t)$ and $u_2(t)$ are finite within the domain \eq{domain}, which we consider as the domain of $u(t)$.

Now we consider 1-qubit states that provide inputs to \eq{NODE2}. Given a sufficiently small $\epsilon>0$, we first define $\theta \in (0,\pi/4)$ by
\begin{equation}
2\sin^2\frac{\theta}{2} = \epsilon.
\label{eq:theta}
\end{equation}
We then construct two 1-qubit states with overlap $1-\epsilon$, namely
\begin{equation}
|\phi(0)\rangle = \frac{1}{\sqrt{2}} (|0\rangle + |1\rangle)
\label{eq:phi_0}
\end{equation}
and
\begin{equation}
|\psi(0)\rangle = \cos\Bigl(\theta+\frac{\pi}{4}\Bigr) |0\rangle + \sin\Bigl(\theta+\frac{\pi}{4}\Bigr) |1\rangle.
\label{eq:psi_0}
\end{equation}
Then the overlap between the two initial states is
\begin{equation}
\langle \phi(0) | \psi(0) \rangle = \cos\theta
= 1-\epsilon.
\label{eq:overlap_0}
\end{equation}
The initial overlap \eq{overlap_0} is larger than the target overlap $3/\sqrt{10}$ in \lem{discrimination} provided $\epsilon<1-3/\sqrt{10}$. For simplicity, we denote
\begin{equation}
\begin{aligned}
v_0 \coloneqq \cos\Bigl(\theta+\frac{\pi}{4}\Bigr), \\
w_0 \coloneqq \sin\Bigl(\theta+\frac{\pi}{4}\Bigr),
\end{aligned}
\label{eq:v0w0}
\end{equation}
and let $v(t)$ and $w(t)$ denote solutions of \eq{NODE2} with initial conditions $v(0)=v_0$ and $w(0)=w_0$, respectively. Since $w_0 > 1/\Rq$, we see that $w(t)$ increases with $t$, satisfying
\begin{equation}
\frac{1}{\Rq} \le \frac{1}{\sqrt{2}} < w_0 < w(t),
\label{eq:vbw}
\end{equation}
and
\begin{equation}
v(t) < w(t)
\label{eq:vw}
\end{equation}
for any time $0<t<t^{\ast}$, whatever the behavior of $v(t)$.

We now study the outputs of our problem. For the state $|\phi(0)\rangle$, the initial condition for \eq{NODE2} is $[1/\sqrt{2} ; 1/\sqrt{2}]$. Thus, the output for any $t \ge 0$ is
\begin{equation}
|\phi(t)\rangle = \frac{1}{\sqrt{2}} (|0\rangle + |1\rangle).
\label{eq:phi_2}
\end{equation}

For the state $|\psi(0)\rangle$, the initial condition for \eq{NODE2} is $[v_0;w_0]$. We now discuss how to select a terminal time $T$ to give a useful output state $|\psi(T)\rangle$. For simplicity, we denote the ratio of $w(t)$ and $v(t)$ by
\begin{equation}
K(t) \coloneqq \frac{w(t)}{v(t)}.
\label{eq:wkv}
\end{equation}
Noticing that $w(t)$ goes to infinity as $t$ approaches $t^{\ast}$, while $v(t)$ remains finite within \eq{domain}, there exists a terminal time $T$ such that\footnote{More concretely, we take $v_{\text{max}} = \max\{v_0, v(t^{\ast})\}$ that upper bounds $v(t)$ on the domain $[0,t^{\ast})$, in which $v(t^{\ast})$ is a finite value since $v_0<w_0$. Then there exists a terminal time $T$ such that $w(T)=2v_{\text{max}}$, and hence $K(T) = w(T)/v(T) \ge 2$.
}
\begin{equation}
K(T) \ge 2.
\label{eq:T2}
\end{equation}
The normalized output state at this time $T$ is
\begin{equation}
|\psi(T)\rangle = \frac{1}{\sqrt{K(T)^2+1}} (|0\rangle + K(T)|1\rangle).
\label{eq:psi_2}
\end{equation}

Combining \eq{phi_2} with \eq{psi_2}, the overlap of $|\phi(T)\rangle$ and $|\psi(T)\rangle$ is
\begin{equation}
\langle \phi(T) | \psi(T) \rangle = \frac{K(T)+1}{\sqrt{2K(T)^2+2}} 
\le \frac{3}{\sqrt{10}}
\label{eq:overlap}
\end{equation}
using \eq{T2}.

Finally, we estimate the evolution time $T$, which is implicitly defined by \eq{T2}. We can upper bound its value by $t^{\ast}$. According to \eq{domain}, we have
\begin{equation}
T < t^{\ast} = \log\biggl(\frac{\Rq}{\Rq-\frac{1}{w_0}}\biggr) < \log\biggl(\frac{\sqrt{2}}{\sqrt{2}-\frac{1}{w_0}}\biggr)
\end{equation}
since the function $\log(x/(x-c))$ decreases monotonically with $x$ for $x>c>0$.
Using \eq{overlap_0} to rewrite this expression in terms of $\epsilon$, we have
\begin{equation}
T < t^{\ast} < \log\Biggl(\frac{\sqrt{2}}{\sqrt{2}-\frac{1}{\sin(\theta+\frac{\pi}{4})}}\Biggr)
= \log\biggl(1 + \frac{1}{\sqrt{2\epsilon-\epsilon^2}-\epsilon}\biggr),
\label{eq:T_estimate}
\end{equation}
which scales like $\frac{1}{2} \log(1/2\epsilon)$ as $\epsilon \to 0$. Therefore $T = O(\log({1}/{\epsilon}))$
as claimed.
\end{proof}

\subsection{Proof of \texorpdfstring{\thm{lower}}{Theorem \ref{thm:lower}}}

We now establish our main lower bound result.

\begin{proof}
As introduced in the proof of \lem{discrimination}, consider the quadratic ODE \eq{NODE2}; the two initial states of a qubit $|\phi(0)\rangle$ and $|\psi(0)\rangle$ defined in \eq{phi_0} and \eq{psi_0}, respectively; and the terminal time $T$ defined in \eq{T2}.

Suppose we have a quantum algorithm that, given a black box to prepare a state that is either $|\phi(0)\rangle$ or $|\psi(0)\rangle$, can produce quantum states $|\phi^\prime(T)\rangle$ or $|\psi^\prime(T)\rangle$ that are within distance $\delta$ of $|\phi(T)\rangle$ and $|\psi(T)\rangle$, respectively. Since by \lem{discrimination}, $|\phi(T)\rangle$ and $|\psi(T)\rangle$ have constant overlap, the overlap between $|\phi^\prime(T)\rangle$ and $|\psi^\prime(T)\rangle$ is also constant for sufficiently small $\delta$. More precisely, we have
\begin{equation}
    \langle \phi(T) | \psi(T) \rangle \le \frac{3}{\sqrt{10}}
\end{equation}
by \eq{overlap}, which implies
\begin{equation}
    \| |\phi(T)\rangle - \ket{\psi(T)}\| \geq \sqrt{2\biggl(1-\frac{3}{\sqrt{10}}\biggr)} > 0.32.
\end{equation}
We also have
\begin{equation}
    \|\ket{\phi(T)} - \ket{\phi^\prime(T)}\|\leq \delta,
\end{equation}
and similarly for $\psi(T)$. These three inequalities give us
\begin{align}
    \|\ket{\phi^\prime(T)} - \ket{\psi^\prime(T)}\|
    &=\|(\ket{\phi(T)} - \ket{\psi(T))} - (\ket{\phi(T)} - \ket{\phi^\prime(T))} - (\ket{\psi^\prime(T)} - \ket{\psi(T))}\| \nonumber\\
    &\geq\|(\ket{\phi(T)} - \ket{\psi(T))}\| - \|(\ket{\phi(T)} - \ket{\phi^\prime(T))}\| - \|(\ket{\psi^\prime(T)} - \ket{\psi(T))}\| \nonumber\\
    &> 0.32 - 2\delta,
\end{align}
which is at least a constant for (say) $\delta<0.15$.

\lem{search} therefore shows that preparing the states $|\phi'(T)\rangle$ and $|\psi'(T)\rangle$ requires time $\Omega(1/\epsilon)$, as these states can be used to distinguish the two possibilities with bounded error. By \lem{discrimination}, this time is $2^{\Omega(T)}$. This shows that we need at least exponential simulation time to approximate the solution of arbitrary quadratic ODEs to within sufficiently small bounded error when $\Rq\geq \sqrt{2}$. 
\end{proof}

Note that exponential time is achievable since our QCL algorithm can solve the problem by taking $N$ to be exponential in $T$, where $N$ is the truncation level of Carleman linearization. (The algorithm of Leyton and Osborne also solves quadratic differential equations with complexity exponential in $T$, but requires the additional assumptions that the quadratic polynomial is measure-preserving and Lipschitz continuous \cite{LO08}.)

\section{Applications}
\label{sec:application}

Due to the assumptions of our analysis, our quantum Carleman linearization algorithm can only be applied to problems with certain properties. First, there are two requirements to guarantee convergence of the inhomogeneous Carleman linearization: the system must have linear dissipation, manifested by $\Re(\lambda_1) < 0$; and the dissipation must be sufficiently stronger than both the nonlinear and the forcing terms, so that $\Rq < 1$. Dissipation typically leads to an exponentially decaying solution, but for the dependency on $g$ and $q$ in \eq{query_complexity} to allow an efficient algorithm, the solution cannot exponentially approach zero.

However, this issue does not arise if the forcing term $F_0$ resists the exponential decay towards zero, instead causing the solution to decay towards some non-zero (possibly time-dependent) state. The norm of the state that is exponentially approached can possibly decay towards zero, but this decay itself must happen slower than exponentially for the algorithm to be efficient.\footnote{Also note that the QCL algorithm might provide an advantage over classical computation for homogeneous equations in cases where only evolution for a short time is of interest.} 

We now investigate possible applications that satisfy these conditions. First we present an application governed by ordinary differential equations, and then we present possible applications in partial differential equations.

Several physical systems can be represented in terms of quadratic ODEs. Examples include models of interacting populations of predators and prey \cite{Wal83}, dynamics of chemical reactions \cite{Mon72,CNF19}, and the spread of an epidemic \cite{BC12}. We now give an example of the latter, based on the epidemiological model used in \cite{WLX20} to describe the early spread of the COVID-19 virus.

The so-called SEIR model divides a population of $P$ individuals into four components: susceptible ($P_S$), exposed ($P_E$), infected ($P_I$), and recovered ($P_R$). We denote the rate of transmission from an infected to a susceptible person by $\rtra$, the typical time until an exposed person becomes infectious by the latent time $\Tlat$, and the typical time an infectious person can infect others by  the infectious time $\Tinf$. Furthermore we assume that there is a flux $\Lambda$ of individuals constantly refreshing the population. This flux corresponds to susceptible individuals moving into, and individuals of all components moving out of, the population, in such a way that the total population remains constant.

To ensure that there is sufficiently strong linear decay to guarantee Carleman convergence, we also add a vaccination term to the $P_S$ component. We choose an approach similar to that of \cite{GYI08}, but denote the vaccination rate, which is approximately equal to the fraction of susceptible individuals vaccinated each day, by $\rvac$. The model is then
\begin{align}
    \frac{\d P_S}{\d t} &= - \Lambda \frac{P_S}{P} - \rvac P_S  + \Lambda - \rtra P_S \frac{P_I}{P}\label{eq:seir_dSdt}\\
    \frac{\d P_E}{\d t} &= - \Lambda \frac{P_E}{P} - \frac{P_E}{\Tlat} + \rtra P_S \frac{P_I}{P}\label{eq:seir_dEdt} \\
    \frac{\d P_I}{\d t} &= - \Lambda \frac{P_I}{P} + \frac{P_E}{\Tlat} - \frac{P_I}{\Tinf}\label{eq:seir_dIdt}\\
    \frac{\d P_R}{\d t} &=- \Lambda \frac{P_R}{P} + \rvac P_S + \frac{P_I}{\Tinf}\label{eq:seir_dVdt}. 
\end{align}

The sum of equations \eq{seir_dSdt}--\eq{seir_dVdt} shows that $P = P_S+P_E+P_I+P_R$ is a constant. Hence we do not need to include the equation for $P_R$ in our analysis, which is crucial since the $P_R$ component would have introduced positive eigenvalues. The matrices corresponding to \eq{NODE} are then
\begin{align}
  F_0 &= \begin{pmatrix} 
  \Lambda\\
  0\\
  0\\
  \end{pmatrix}, \quad
  F_1 = \begin{pmatrix} 
  -\frac{\Lambda}{P} - \rvac & 0 & 0 \\
  0 & -\frac{\Lambda}{P} - \frac{1}{\Tlat} & 0\\
  0 &  \frac{1}{\Tlat} & -\frac{\Lambda}{P} - \frac{1}{\Tinf}\\
  \end{pmatrix}, \\
\setcounter{MaxMatrixCols}{20}
  F_2 &= \begin{pmatrix} 
  0 & 0 & -\frac{\rtra}{P} & 0 & 0 & 0 & 0 & 0 & 0\\
  0 & 0 & \frac{\rtra}{P} & 0 & 0 & 0 & 0 & 0 & 0\\
  0 & 0 & 0 & 0 & 0 & 0 & 0 & 0 & 0\\
  \end{pmatrix}.
\end{align}

Since $F_1$ is a triangular matrix, its eigenvalues are located on its diagonal, so $\Re(\lambda_1) = -\Lambda/P-\min\left\{\rvac, 1/\Tlat, 1/\Tinf \right\}$. Furthermore we can bound $P/\sqrt{3} \le \|u_{\text{in}}\| \le P$, $\|F_0\| = \Lambda$, and $\|F_2\| = \sqrt 2\rtra /P$, so
\begin{align}
\Rq &\le \frac{\sqrt{2}\rtra+\sqrt{3}\Lambda/P}{\min\left\{\rvac, 1/\Tlat, 1/\Tinf \right\}+\Lambda/P}.
\end{align}
We see that the condition for guaranteed convergence of Carleman linearization is $\sqrt{2} \rtra < \min\left\{\rvac, 1/\Tlat, 1/\Tinf\right\} - (\sqrt{3}-1)$ and $\Lambda \le \sqrt 2\rtra /P$. Essentially, the Carleman method only converges if the (nonlinear) transmission is sufficiently slower than the (linear) decay and the travel flux is low compared to the transmission rate divided by the total population.

To assess how restrictive this assumption is, we consider the SEIR parameters used in \cite{WLX20}. Note that they also included separate components for asymptomatic and hospitalized persons, but to simplify the analysis we include both of these components in the $P_I$ component. In their work, they considered a city with approximately $P = 10^7$ inhabitants. In a specific period, they estimated a travel flux%
\footnote{
For negligible travel flux $\Lambda \approx 0$, the solution $u(t)$ will decay to a small value, so the solution is not interesting in the asymptotic limit. However, in practice we are interested in the behavior over shorter times for which the population has not significantly decayed.
} 
of $\Lambda \approx 0$ individuals per day, latent time $\Tlat = 5.2$ days, infectious time $\Tinf = 2.3$ days,  and transmission rate $\rtra \approx 0.13\text{ days}^{-1}$. We let the initial condition be dominated by the susceptible component so that $\|u_{\text{in}}\| \approx P$, and we assume\footnote{This example arguably corresponds to quite rapid vaccination, and is chosen here such that $\Rq$ remains smaller than one, as required to formally guarantee convergence of the Carleman method. However, as shown in the upcoming example of the Burgers equation, larger values of $\Rq$ might still allow for convergence in practice, suggesting that our algorithm might handle lower values of the vaccination rate.} that $\rvac > 1/\Tlat \approx 0.19\text{ days}^{-1}$ and negligible travel flux ($\Lambda=0$).
With the stated parameters, a direct calculation gives $\Rq = 0.956$ and ensures $\Lambda \le \sqrt 2\rtra /P$, showing that the assumptions of our algorithm can correspond to some real-world problems that are only moderately nonlinear.

While the example discussed above has only a constant number of variables, this example can be generalized to a high-dimensional system of ODEs that models the early spread over a large number of cities with interaction, similar to what is done in \cite{BK15} and \cite{QA18}.

Other examples of high-dimensional ODEs arise from the discretization of certain PDEs. Consider, for example, equations for $\vec u(\vec r, t)$ of the type
\begin{equation}
    \partial_t \vec u + (\vec u \cdot \nabla) \vec u + \beta \vec u = \nu \nabla^2 \vec u + \vec f. \label{eq:NS-LD}
\end{equation}
with the forcing term $\vec f$ being a function of both space and time. This equation can be cast in the form of \eq{NODE} by using standard discretizations of space and time. Equations of the form \eq{NS-LD} can represent Navier--Stokes-type equations, which are ubiquitous in fluid mechanics \cite{Lem18}, and related models such as those studied in \cite{KPS90, SZ89, VDF94}
to describe the formation of large-scale structure in the universe. Similar equations also appear in models of magnetohydrodynamics (e.g., \cite{davidson_2001}), or the motion of free particles that stick to each other upon collision \cite{BG98}. In the inviscid case, $\nu = 0$, the resulting Euler-type equations with linear damping are also of interest, both for modeling micromechanical devices \cite{Bao00} and for their intimate connection with viscous models \cite{Daf05}.

\begin{figure}[htbp]
    \centering
    \includegraphics[width=0.85\textwidth]{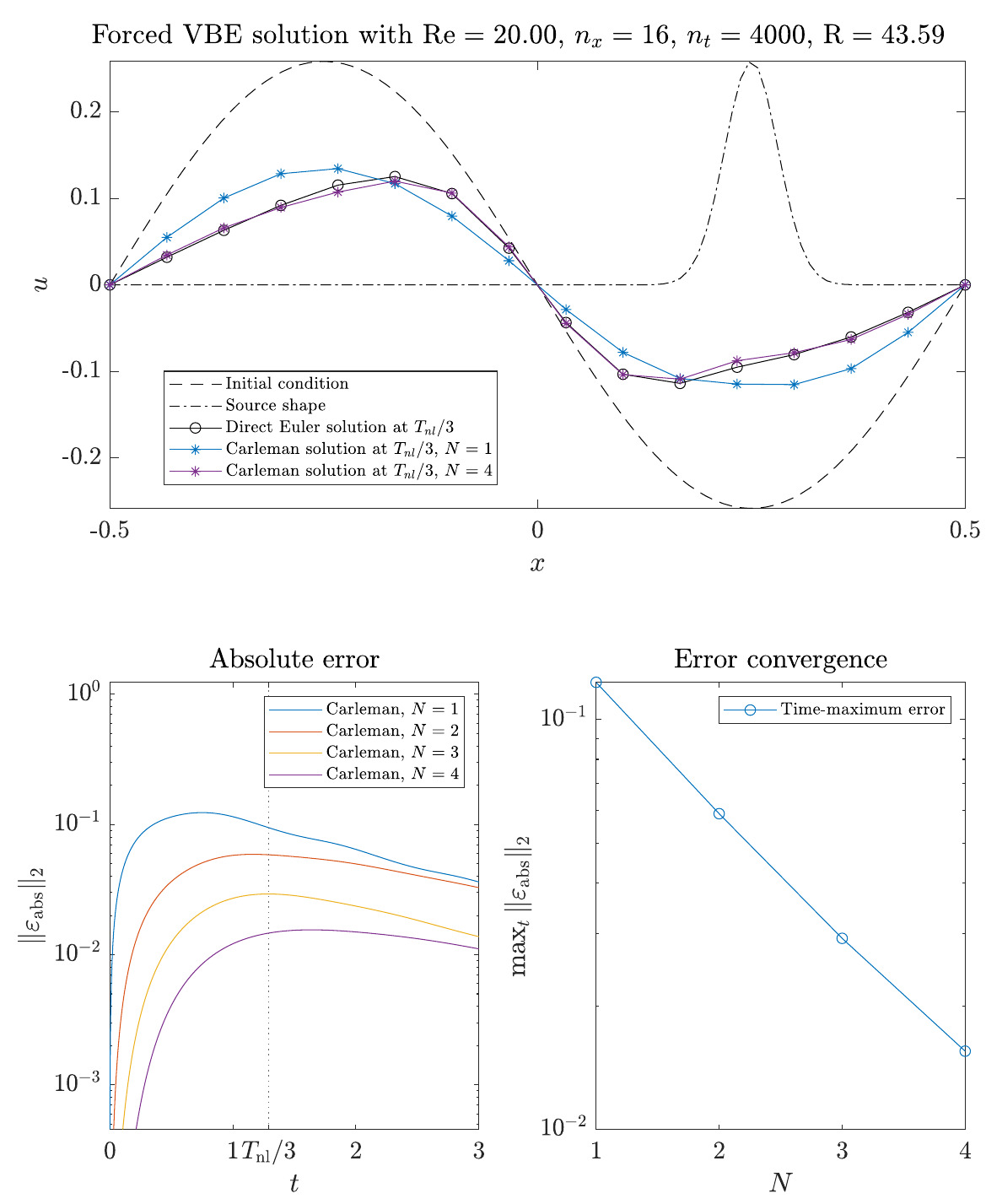}
    \caption{Integration of the forced viscous Burgers equation using Carleman linearization on a classical computer (source code available at \url{https://github.com/hermankolden/CarlemanBurgers}). The viscosity is set so that the Reynolds number $\Re = U_0L_0/\nu = 20$. The parameters $n_x=16$ and $n_t=4000$ are the number of spatial and temporal discretization intervals, respectively. The corresponding Carleman convergence parameter is $\Rq = 43.59$. Top: Initial condition and solution plotted at a third of the nonlinear time $\frac 1 3 T_{\mathrm{nl}} = \frac{L_0}{3U_0}$. Bottom: $l_2$ norm of the absolute error between the Carleman solutions at various truncation levels $N$ (left), and the convergence of the corresponding time-maximum error (right).}
    \label{fig:vbe_re0_20}
\end{figure}

As a specific example, consider the one-dimensional forced viscous Burgers equation
\begin{equation}
    \label{eq:vbe}
    \partial_t u + u \partial_x u = \nu \partial_x^2 u + f,
\end{equation}
which is the one-dimensional case of equation \eq{NS-LD} with $\beta = 0$. Equation \eq{vbe} is often used as a simple model of convective flow \cite{Bur48}. For concreteness, let the initial condition be $u(x, 0) = U_0 \sin(2\pi x/L_0)$ on the domain $x \in [-L_0/2, L_0/2]$, and use Dirichlet boundary conditions $u(-L_0/2,0) \allowbreak = u(L_0/2,0) = 0$. We force this equation using a localized off-center Gaussian with a sinusoidal time dependence,\footnote{Note that this forcing does not satisfy the general conditions for efficient implementation of our algorithm since it is not sparse. However, we expect that the algorithm can still be implemented efficiently for a structured non-sparse forcing term such as in this example.} given by $f(x,t) = U_0\exp\left(-\frac{(x-L_0/4)^2}{2(L_0/32)^2}\right)\cos(2\pi t)$. To solve this equation using the Carleman method, we discretize the spatial domain into $n_x$ points and use central differences for the derivatives to get
\begin{equation}
    \label{eq:disc_vbe}
    \partial_t u_i = \nu \frac{u_{i+1}-2u_i+u_{i-1}}{\Delta x^2} - \frac{u^2_{i+1}-u^2_{i-1}}{4\Delta x} + f_i
\end{equation}
with $\Delta x = L_0 / (n_x - 1)$. This equation is of the form \eq{NODE} and can thus generate the Carleman system \eq{UODE}. The resulting linear ODE can then be integrated using the forward Euler method, as shown in \fig{vbe_re0_20}. In this example, the viscosity $\nu$ is defined such that the Reynolds number $\Re \coloneqq U_0L_0/\nu = 20$, and $n_x=16$ spatial discretization points were sufficient to resolve the solution. 
The figure compares the Carleman solution with the solution obtained via direct integration of \eq{disc_vbe} with the forward Euler method (i.e., without Carleman linearization). By inserting the matrices $F_0$, $F_1$, and $F_2$ corresponding to equation \eq{disc_vbe} into the definition of $\Rq$ \eq{A1}, we find that $\Re(\lambda_1)$ is indeed negative as required, given Dirichlet boundary conditions, but the parameters used in this example result in $\Rq \approx 44$. Even though this does not satisfy the requirement $\Rq < 1$ of the QCL algorithm, we see from the absolute error plot in \fig{vbe_re0_20} that the maximum absolute error over time decreases exponentially as the truncation level $N$ is incremented (in this example, the maximum Carleman truncation level considered is $N=4$). Surprisingly, this suggests that in this example, the error of the classical Carleman method converges exponentially with $N$, even though $\Rq > 1$.

\section{Discussion}
\label{sec:discussion}

In this paper we have presented a quantum Carleman linearization (QCL) algorithm for a class of quadratic nonlinear differential equations. Compared to the previous approach of \cite{LO08}, our algorithm improves the complexity from an exponential dependence on $T$ to a nearly quadratic dependence, under the condition $\Rq < 1$ as defined in \eq{A1}. Qualitatively, this means that the system must be dissipative and that the nonlinear and inhomogeneous effects must be small relative to the linear effects. We have also provided numerical results suggesting the classical Carleman method may work on certain PDEs that do not strictly satisfy the assumption $\Rq < 1$.
Furthermore, we established a lower bound showing that for general quadratic differential equations with $\Rq \ge \sqrt{2}$, quantum algorithms must have worst-case complexity exponential in $T$. 
We also discussed several potential applications arising in biology and fluid and plasma dynamics.

It is natural to ask whether the result of \thm{main} can be achieved with a classical algorithm, i.e., whether the assumption $\Rq < 1$ makes differential equations classically tractable. Clearly a naive integration of the truncated Carleman system \eq{UODE} is not efficient on a classical computer since the system size is $\Theta(n^N)$. But furthermore, it is unlikely that \emph{any} classical algorithm for this problem can run in time polylogarithmic in $n$. If we consider \prb{node} with dissipation that is small compared to the total evolution time, but let the nonlinearity and forcing be even smaller such that $\Rq < 1$, then in the asymptotic limit we have a linear differential equation with no dissipation. Hence any classical algorithm that could solve \prb{node} could also solve non-dissipative linear differential equations, which is a $\text{BQP}$-hard problem even when the dynamics are unitary \cite{Fey85}. In other words, an efficient classical algorithm for this problem would imply efficient classical algorithms for any problem that can be solved efficiently by a quantum computer, which is considered unlikely.

Our upper and lower bounds leave a gap in the range $1 \le \Rq < \sqrt2$, for which we do not know the complexity of the quantum quadratic ODE problem. We hope that future work will close this gap and determine for which $\Rq$ the problem can be solved efficiently by quantum computers in the worst case.

Furthermore, the complexity of our algorithm has nearly quadratic dependence on $T$, namely $T^2\poly(\log T)$. It is unknown whether the complexity for quadratic ODEs must be at least linear or quadratic in $T$. Note that sublinear complexity is impossible in general because of the no-fast-forwarding theorem \cite{BAC07}. However, it should be possible to fast-forward the dynamics in special cases, and it would be interesting to understand the extent to which dissipation enables this.

The complexity of our algorithm depends on the parameter $q$ defined in \thm{main}, which characterizes the decay of the final solution relative to the initial condition. This restricts the utility of our result, since we must have a suitable initial condition and terminal time such that the final state is not exponentially smaller than the initial state. However, it is unlikely that such a dependence can be significantly improved, since renormalization of the state can be used to implement postselection, which would imply the unlikely consequence $\mathrm{BQP} = \mathrm{PP}$ (see Section 8 of \cite{BCOW17} for further discussion). As discussed in the introduction, the solution of a homogeneous dissipative equation necessarily decays exponentially in time, so our method is not asymptotically efficient.  However, for inhomogeneous equations the final state need not be exponentially smaller than the initial state even in a long-time simulation, suggesting that our algorithm could be especially suitable for models with forcing terms.

It is possible that variations of the Carleman linearization procedure could increase the accuracy of the result. For instance, instead of using just tensor powers of $u$ as auxiliary variables, one could use other nonlinear functions. Several previous papers on Carleman linearization have suggested using multidimensional orthogonal polynomials \cite{BR63,FP17}. They also discuss approximating higher-order terms with lower-order ones in \eq{UODE} instead of simply dropping them, possibly improving accuracy. Such changes would however alter the structure of the resulting linear ODE, which could affect the quantum implementation.

The quantum part of the algorithm might also be improved. In this paper we limit ourselves to the first-order Euler method to discretize the linearized ODEs in time. This is crucial for the analysis in \lem{global}, which states the global error increases at most linearly with $T$. To establish this result for the Euler method, it suffices to choose the time step \eq{h} to ensure $\|I+Ah\|\le1$, and then estimate the growth of global error by \eq{global_0}. However, it is unclear how to give a similar bound for higher-order numerical schemes. If this obstacle could be overcome, the error dependence of the complexity might be improved.

It is also natural to ask whether our approach can be improved by taking features of particular systems into account. Since the Carleman method has only received limited attention and has generally been used for purposes other than numerical integration, it seems likely that such improvements are possible.
In fact, the numerical results discussed in \sec{application} (see in particular \fig{vbe_re0_20}) suggest that the condition $\Rq < 1$ is not a strict requirement for the viscous Burgers equation, since we observe convergence even though $\Rq \approx 44$. This suggests that some property of equation \eq{vbe} makes it more amenable to Carleman linearization than our current analysis predicts. We leave a detailed investigation of this for future work.

A related question is whether our algorithm can efficiently simulate systems exhibiting dynamical chaos. The condition $\Rq<1$ might preclude chaos, but we do not have a proof of this. More generally, the presence or absence of chaos might provide a more fine-grained picture of the algorithm's efficiency.

When contemplating applications, it should be emphasized that our approach produces a state vector that encodes the solution without specifying how information is to be extracted from that state. Simply producing a state vector is not enough for an end-to-end application since the full quantum state cannot be read out efficiently. In some cases in may be possible to extract useful information by sampling a simple observable, whereas in other cases, more sophisticated postprocessing may be required to infer a desired property of the solution. Our method does not address this issue, but can be considered as a subroutine whose output will be parsed by subsequent quantum algorithms. 
We hope that future work will address this issue and develop end-to-end applications of these methods.

Finally, the algorithm presented in this paper might be extended to solve related mathematical problems on quantum computers. Obvious candidates include initial value problems with time-dependent coefficients and boundary value problems. Carleman methods for such problems are explored in \cite{KS91}, but it is not obvious how to implement those methods in a quantum algorithm. It is also possible that suitable formulations of problems in nonlinear optimization or control could be solvable using related techniques.

\section*{Acknowledgments}

We thank Paola Cappellaro for valuable discussions and comments. JPL thanks Aaron Ostrander for inspiring discussions. HØK thanks Bernhard Paus Græsdal for introducing him to the Carleman method. HKK and NFL thank Seth Lloyd for preliminary discussions on nonlinear equations and the quantum linear systems algorithm. We also thank Dong An and anonymous referees for pointing out the exponential decay of solutions to dissipative homogeneous equations. AMC and JPL did part of this work while visiting the Simons Institute for the Theory of Computing in Berkeley and gratefully acknowledge its hospitality. 

Furthermore, we thank Dong An for pointing out a technical issue with \lem{Carleman}, and therefore \thm{main}, in an earlier version of this paper. This revised version adds the condition $\|F_0\| \le \|F_2\|$, which is sufficient for the result to hold.

AMC and JPL acknowledge support from the Department of Energy, Office of Science, Office of Advanced Scientific Computing Research, Quantum Algorithms Teams and Accelerated Research in Quantum Computing programs, and from the National Science Foundation (CCF-1813814).
HØK, HKK, and NFL were partially funded by award no.\ DE-SC0020264 from the Department of Energy.
\bibliographystyle{myhamsplain}
\bibliography{qplasma}

\providecommand{\bysame}{\leavevmode\hbox to3em{\hrulefill}\thinspace}
\begin{thebibliography}{10}

\bibitem{Aar05}
Scott Aaronson, \emph{{NP}-complete problems and physical reality}, ACM SIGACT News \textbf{36} (2005), no.~1, 30--52, \href{https://arxiv.org/abs/quant-ph/0502072}{arXiv:quant-ph/0502072}.

\bibitem{AL98}
Daniel~S. Abrams and Seth Lloyd, \emph{Nonlinear quantum mechanics implies polynomial-time solution for {NP}-complete and \#{P} problems}, Physical Review Letters \textbf{81} (1998), no.~18, 3992, \href{https://arxiv.org/abs/quant-ph/9801041}{arXiv:quant-ph/9801041}.

\bibitem{And82}
Roberto F.~S. Andrade, \emph{Carleman embedding and {L}yapunov exponents}, Journal of Mathematical Physics \textbf{23} (1982), no.~12, 2271--2275.

\bibitem{Atk08}
Kendall~E. Atkinson, \emph{An introduction to numerical analysis}, John Wiley \& Sons, 2008.

\bibitem{Bao00}
Min-Hang Bao, \emph{Micro mechanical transducers: pressure sensors, accelerometers and gyroscopes}, Elsevier, 2000.

\bibitem{BR63}
Richard Bellman and John~M. Richardson, \emph{On some questions arising in the approximate solution of nonlinear differential equations}, Quarterly of Applied Mathematics \textbf{20} (1963), 333--339.

\bibitem{Ber14}
Dominic~W. Berry, \emph{High-order quantum algorithm for solving linear differential equations}, Journal of Physics A: Mathematical and Theoretical \textbf{47} (2014), no.~10, 105301, \href{https://arxiv.org/abs/1010.2745}{arXiv:1010.2745}.

\bibitem{BAC07}
Dominic~W. Berry, Graeme Ahokas, Richard Cleve, and Barry~C. Sanders, \emph{Efficient quantum algorithms for simulating sparse {H}amiltonians}, Communications in Mathematical Physics \textbf{270} (2007), 359–371, \href{https://arxiv.org/abs/quant-ph/0508139}{arXiv:quant-ph/0508139}.

\bibitem{BCOW17}
Dominic~W. Berry, Andrew~M. Childs, Aaron Ostrander, and Guoming Wang, \emph{Quantum algorithm for linear differential equations with exponentially improved dependence on precision}, Communications in Mathematical Physics \textbf{356} (2017), no.~3, 1057--1081, \href{https://arxiv.org/abs/1701.03684}{arXiv:1701.03684}.

\bibitem{BK15}
Derdei Bichara, Yun Kang, Carlos Castillo-Chávez, Richard Horan, and Charles Perrings, \emph{{SIS} and {SIR} epidemic models under virtual dispersal}, 03 2015.

\bibitem{BC12}
Fred Brauer and Carlos Castillo-Chavez, \emph{Mathematical models in population biology and epidemiology}, vol.~2, Springer, 2012.

\bibitem{BG98}
Yann Brenier and Emmanuel Grenier, \emph{Sticky particles and scalar conservation laws}, SIAM Journal on Numerical Analysis \textbf{35} (1998), no.~6, 2317--2328.

\bibitem{Bro14}
Roger Brockett, \emph{The early days of geometric nonlinear control}, Automatica \textbf{50} (2014), no.~9, 2203--2224.

\bibitem{Bur48}
Johannes~M. Burgers, \emph{A mathematical model illustrating the theory of turbulence}, Advances in Applied Mechanics, vol.~1, Elsevier, 1948, pp.~171--199.

\bibitem{CPPTK13}
Yudong Cao, Anargyros Papageorgiou, Iasonas Petras, Joseph Traub, and Sabre Kais, \emph{Quantum algorithm and circuit design solving the {P}oisson equation}, New Journal of Physics \textbf{15} (2013), no.~1, 013021, \href{https://arxiv.org/abs/1207.2485}{arXiv:1207.2485}.

\bibitem{Car32}
Torsten Carleman, \emph{Application de la th{\'e}orie des {\'e}quations int{\'e}grales lin{\'e}aires aux syst{\`e}mes d'{\'e}quations diff{\'e}rentielles non lin{\'e}aires}, Acta Mathematica \textbf{59} (1932), no.~1, 63--87.

\bibitem{CNF19}
Alessandro Ceccato, Paolo Nicolini, and Diego Frezzato, \emph{Recasting the mass-action rate equations of open chemical reaction networks into a universal quadratic format}, Journal of Mathematical Chemistry \textbf{57} (2019), 1001--1018.

\bibitem{CKS15}
Andrew~M. Childs, Robin Kothari, and Rolando~D. Somma, \emph{Quantum algorithm for systems of linear equations with exponentially improved dependence on precision}, SIAM Journal on Computing \textbf{46} (2017), no.~6, 1920--1950, \href{https://arxiv.org/abs/1511.02306}{arXiv:1511.02306}.

\bibitem{CL19}
Andrew~M. Childs and Jin-Peng Liu, \emph{Quantum spectral methods for differential equations}, Communications in Mathematical Physics \textbf{375} (2020), 1427--1457, \href{https://arxiv.org/abs/1901.00961}{arXiv:1901.00961}.

\bibitem{CLO20}
Andrew~M. Childs, Jin-Peng Liu, and Aaron Ostrander, \emph{High-precision quantum algorithms for partial differential equations}, 2020, \href{https://arxiv.org/abs/2002.07868}{arXiv:2002.07868}.

\bibitem{CY16}
Andrew~M. Childs and Joshua Young, \emph{Optimal state discrimination and unstructured search in nonlinear quantum mechanics}, Physical Review A \textbf{93} (2016), no.~2, 022314, \href{https://arxiv.org/abs/1507.06334}{arXiv:1507.06334}.

\bibitem{CJS13}
B.~David Clader, Bryan~C. Jacobs, and Chad~R. Sprouse, \emph{Preconditioned quantum linear system algorithm}, Physical Review Letters \textbf{110} (2013), no.~25, 250504, \href{https://arxiv.org/abs/1301.2340}{arXiv:1301.2340}.

\bibitem{CJO19}
Pedro C.~S. Costa, Stephen Jordan, and Aaron Ostrander, \emph{Quantum algorithm for simulating the wave equation}, Physical Review A \textbf{99} (2019), no.~1, 012323, \href{https://arxiv.org/abs/1711.05394}{arXiv:1711.05394}.

\bibitem{Daf05}
Constantine~M. Dafermos, \emph{Hyperbolic conservation laws in continuum physics}, vol.~3, Springer, 2005.

\bibitem{Dah63}
Germund~G. Dahlquist, \emph{A special stability problem for linear multistep methods}, BIT Numerical Mathematics \textbf{3} (1963), no.~1, 27--43.

\bibitem{davidson_2001}
P.~A. Davidson, \emph{An introduction to magnetohydrodynamics}, Cambridge Texts in Applied Mathematics, Cambridge University Press, 2001.

\bibitem{DS20}
Ilya~Y. Dodin and Edward~A. Startsev, \emph{On applications of quantum computing to plasma simulations}, 2020, \href{https://arxiv.org/abs/2005.14369}{arXiv:2005.14369}.

\bibitem{ESP19}
Alexander Engel, Graeme Smith, and Scott~E. Parker, \emph{Quantum algorithm for the {V}lasov equation}, Physical Review A \textbf{100} (2019), no.~6, 062315, \href{https://arxiv.org/abs/1907.09418}{arXiv:1907.09418}.

\bibitem{Fey85}
Richard~P. Feynman, \emph{Quantum mechanical computers}, Optics News \textbf{11} (1985), 11--20.

\bibitem{FP17}
Marcelo Forets and Amaury Pouly, \emph{Explicit error bounds for {C}arleman linearization}, 2017, \href{https://arxiv.org/abs/1711.02552}{arXiv:1711.02552}.

\bibitem{GMP05}
Alfredo Germani, Costanzo Manes, and Pasquale Palumbo, \emph{Filtering of differential nonlinear systems via a {C}arleman approximation approach}, Proceedings of the 44th IEEE Conference on Decision and Control, pp.~5917--5922, 2005.

\bibitem{HHL09}
Aram~W. Harrow, Avinatan Hassidim, and Seth Lloyd, \emph{Quantum algorithm for linear systems of equations}, Physical Review Letters \textbf{103} (2009), no.~15, 150502, \href{https://arxiv.org/abs/0811.3171}{arXiv:0811.3171}.

\bibitem{Hel69}
Carl~W. Helstrom, \emph{Quantum detection and estimation theory}, Journal of Statistical Physics \textbf{1} (1969), 231--252.

\bibitem{Jos20}
Ilon Joseph, \emph{Koopman-von {N}eumann approach to quantum simulation of nonlinear classical dynamics}, 2020, \href{https://arxiv.org/abs/2003.09980}{arXiv:2003.09980}.

\bibitem{Ker81}
Edward~H. Kerner, \emph{Universal formats for nonlinear ordinary differential systems}, Journal of Mathematical Physics \textbf{22} (1981), no.~7, 1366--1371.

\bibitem{KPS90}
Lev Kofman, Dmitri Pogosian, and Sergei Shandarin, \emph{Structure of the universe in the two-dimensional model of adhesion}, Monthly Notices of the Royal Astronomical Society \textbf{242} (1990), no.~2, 200--208.

\bibitem{KS91}
Krzysztof Kowalski and Willi-Hans Steeb, \emph{Nonlinear dynamical systems and {C}arleman linearization}, World Scientific, 1991.

\bibitem{Lem18}
Pierre~Gilles Lemari{\'e}-Rieusset, \emph{The {N}avier-{S}tokes problem in the 21st century}, CRC Press, 2018.

\bibitem{LO08}
Sarah~K. Leyton and Tobias~J. Osborne, \emph{A quantum algorithm to solve nonlinear differential equations}, 2008, \href{https://arxiv.org/abs/0812.4423}{arXiv:0812.4423}.

\bibitem{LMS20}
Noah Linden, Ashley Montanaro, and Changpeng Shao, \emph{Quantum vs.\ classical algorithms for solving the heat equation}, \href{https://arxiv.org/abs/2004.06516}{arXiv:2004.06516}.

\bibitem{LPG20}
Seth Lloyd, Giacomo De~Palma, Can Gokler, Bobak Kiani, Zi-Wen Liu, Milad Marvian, Felix Tennie, and Tim Palmer, \emph{Quantum algorithm for nonlinear differential equations}, 2020, \href{https://arxiv.org/abs/2011.06571}{arXiv:2011.06571}.

\bibitem{LT87}
Gerasimos Lyberatos and Christos~A. Tsiligiannis, \emph{A linear algebraic method for analysing {H}opf bifurcation of chemical reaction systems}, Chemical Engineering Science \textbf{42} (1987), no.~5, 1242--1244.

\bibitem{MP16}
Ashley Montanaro and Sam Pallister, \emph{Quantum algorithms and the finite element method}, Physical Review A \textbf{93} (2016), no.~3, 032324, \href{https://arxiv.org/abs/1512.05903}{arXiv:1512.05903}.

\bibitem{Mon72}
Elliott~W. Montroll, \emph{On coupled rate equations with quadratic nonlinearities}, Proceedings of the National Academy of Sciences \textbf{69} (1972), no.~9, 2532--2536.

\bibitem{QA18}
Rany Qurratu~Aini, Deden Aldila, and Kiki Sugeng, \emph{Basic reproduction number of a multi-patch {SVI} model represented as a star graph topology}, 10 2018, p.~020237.

\bibitem{RMA09}
Andreas Rauh, Johanna Minisini, and Harald Aschemann, \emph{Carleman linearization for control and for state and disturbance estimation of nonlinear dynamical processes}, IFAC Proceedings Volumes \textbf{42} (2009), no.~13, 455--460.

\bibitem{SZ89}
Sergei~F. Shandarin and Yakov~B. Zeldovich, \emph{The large-scale structure of the universe: turbulence, intermittency, structures in a self-gravitating medium}, Reviews of Modern Physics \textbf{61} (1989), no.~2, 185--220.

\bibitem{SBM06}
Vivek~V. Shende, Stephen~S. Bullock, and Igor~L. Markov, \emph{Synthesis of quantum-logic circuits}, IEEE Transactions on Computer-Aided Design of Integrated Circuits and Systems \textbf{25} (2006), no.~6, 1000--1010, \href{https://arxiv.org/abs/quant-ph/0406176}{arXiv:quant-ph/0406176}.

\bibitem{SW80}
Willi-Hans Steeb and F.~Wilhelm, \emph{Non-linear autonomous systems of differential equations and {C}arleman linearization procedure}, Journal of Mathematical Analysis and Applications \textbf{77} (1980), no.~2, 601--611.

\bibitem{VDF94}
Massimo Vergassola, B{\'e}reng{\`e}re Dubrulle, Uriel Frisch, and Alain Noullez, \emph{Burgers' equation, devil's staircases and the mass distribution for large-scale structures}, Astronomy and Astrophysics \textbf{289} (1994), 325--356.

\bibitem{Wal83}
Paul Waltman, \emph{Competition models in population biology}, SIAM, 1983.

\bibitem{WLX20}
Chaolong Wang, Li~Liu, Xingjie Hao, Huan Guo, Qi~Wang, Jiao Huang, Na~He, Hongjie Yu, Xihong Lin, An~Pan, Sheng Wei, and Tangchun Wu, \emph{Evolving epidemiology and impact of non-pharmaceutical interventions on the outbreak of coronavirus disease 2019 in {W}uhan, {C}hina}, Journal of the American Medical Association \textbf{323} (2020), no.~19, 1915--1923.

\bibitem{GYI08}
Gul Zaman, Yong {Han Kang}, and Il~Hyo Jung, \emph{Stability analysis and optimal vaccination of an {SIR} epidemic model}, Biosystems \textbf{93} (2008), no.~3, 240--249.

\end{thebibliography}

\end{document}